\documentclass[journal,12pt,onecolumn]{IEEEtran}

\usepackage{amsmath,amsfonts}
\usepackage{amssymb,verbatim,amsfonts,amsmath,amsthm}
\usepackage{geometry}
\usepackage{cite}
\usepackage{hyperref}
\usepackage{bm}

\usepackage{tikz-cd}

\usepackage{longtable}
\usepackage{booktabs}
\usepackage{graphicx}
\usepackage[figuresright]{rotating}

\usepackage{makecell}
\usepackage{seqsplit} 

\newtheorem{theorem}{Theorem}
\newtheorem{lemma}{Lemma}

\newtheorem{proposition}{Proposition}
\newtheorem{example}{Example}
\newtheorem{remark}{Remark}
\newtheorem{definition}{Definition}

\begin{document}

\title{Locally Repairable Codes with Availability via Elliptic Function Fields}
\author{Junjie Huang and Chang-An Zhao$^\dagger$
\thanks{J. Huang is with the Department of Mathematics, School of Mathematics, Sun Yat-sen University, Guangzhou 510275, P.R.China. (e-mail: huangjj76@mail2.sysu.edu.cn).}
\thanks{C.-A. Zhao is with the School of Mathematics, Sun Yat-sen University, Guangzhou 510275, P.R.China, and also with the Guangdong Key Laboratory of Information Security, Guangzhou 510006, P.R.China (e-mail: zhaochan3@mail.sysu.edu.cn).}
\thanks{\(\dagger\) Corresponding author. }
}


\maketitle

\begin{abstract}
Locally repairable codes with availability have become essential components in modern large-scale distributed cloud storage systems and numerous other applications. In this paper, we focus on the construction of locally repairable codes with one or two recovering sets via elliptic function fields. Prior pioneering work by Li et al. (IEEE Trans. Inf. Theory, vol. 65, no. 1, 2019) and Ma and Xing (J. Comb. Theory Ser. A., vol. 193, 2023) employed maximal supersingular elliptic curves to obtain several optimal (classical) locally repairable codes. In contrast, we consider ordinary elliptic curves with many rational points. This approach yields several new families  of \(q\)-ary optimal locally repairable codes with length \(O(q+2\sqrt{q})\) and flexible locality. Consequently, our work broadens the selection of curves available for the construction of optimal locally repairable codes.

Furthermore, we present a general framework for constructing locally repairable codes with two recovering sets via automorphism groups of elliptic function fields. To realize this framework, we devise a novel construction for determining the functions \(e_i\) in the construction of locally repairable codes. By employing both supersingular and ordinary elliptic curves, we obtain several families of locally repairable codes with two recovering sets.  In particular, we construct a family of \(q^2\)-ary locally repairable codes with two recovering sets, achieving length \(O(q^2+2q)\) and Singleton-defect \(O\!\left(\frac{2\ell}{q^2+2q-8\ell}\right)\), where \(\ell \mid\mid q + 2\) with \(4\ell < q\).
\end{abstract}


\section{Introduction}

Locally repairable codes with one or more recovering sets have been extensively studied due to recent applications in distributed storage and cloud storage systems \cite{4557599},\cite{On_the_Locality_of_Codeword_Symbols},\cite{6818438},\cite{FORBES201478},\cite{8493531}. With local recovery techniques, we can repair lost encoded data using a local approach. This means we only need a little bit of data, rather than all the information in a codeword. Specifically, in a distributed storage system, when the data in one node is erased, we want to use data from a small set of other nodes to repair the data in the failure node. However, if some nodes in this small set are unavailable, we need to find an alternative set to perform the repair. To overcome this problem, we can work with \(t\) disjoint recovering sets for each node. The number \(t\) of the disjoint sets is called availability. Informally, a block code is said to have availability $t$ and localities \((r_1, r_2, \ldots, r_t)\) if, for every coordinate of a codeword, there exist \(t\) disjoint recovering sets with cardinalities \(r_1, r_2, \ldots, r_t\), such that the value of the coordinate can be recovered by accessing any one of these sets. A formal definition of a locally repairable code with localities \((r_1, r_2, \ldots, r_t)\) is given as follows.
\begin{definition}
    Let \( \mathcal{C} \subseteq \mathbb{F}_q^n \) be a \( q \)-ary block code of length \( n \). For each \( \alpha \in \mathbb{F}_q \) and \( i \in \{1, 2, \dots, n\} \), define  
    \[
        \mathcal{C}(i, \alpha) := \big\{c = (c_1, \dots, c_n) \in \mathcal{C} : c_i = \alpha\big\}.
    \]
    For a subset \( I \subseteq \{1, 2, \dots, n\} \setminus \{i\} \), we denote by \( \mathcal{C}_I(i, \alpha) \) the projection of \( \mathcal{C}(i, \alpha) \) on \( I \). Then \( \mathcal{C} \) is called a locally repairable code with localities \((r_1, r_2, \dots, r_t)\) and availability \( t \) if, for every \( i \in \{1, 2, \dots, n\} \), there exist disjoint subsets \( I_{i,j} \subseteq \{1, 2, \dots, n\} \setminus \{i\} \) for \( j = 1, 2, \dots, t \) with \( |I_{i,j}| \leq r_j \) such that \( \mathcal{C}_{I_{i,j}}(i, \alpha) \) and \( \mathcal{C}_{I_{i,j}}(i, \beta) \) are disjoint for any \( \alpha \neq \beta \).
\end{definition}

In this paper, we focus exclusively on linear locally repairable codes over a finite field \( \mathbb{F}_q \). A linear \(q\)-ary locally repairable code with length \( n \), dimension \( k \), minimum distance \( d \) and localities \((r_1, r_2, \dots, r_t)\) is denoted by \([n, k, d; (r_1, r_2, \cdots, r_t)]_q\). For $t = 1$, locally repairable codes defined above are classical locally repairable codes, i.e. with only one recovering set. It was established in \cite{On_the_Locality_of_Codeword_Symbols} that the minimum distance of the code satisfies the Singleton-type upper bound
\begin{equation}\label{min_d_bound_one}
    d\le n-k-\bigg\lceil\frac{k}{r}\bigg\rceil+2.
\end{equation}
Any (classical) locally repairable code achieving the bound \eqref{min_d_bound_one} is called optimal. Several upper bounds on the minimum distance of locally repairable codes with availability $t$ have been established in the literature. In \cite{7398062}, Rawat et al. generalized the bound \eqref{min_d_bound_one} to the following upper bound of \([n, k, d; (r, r, \cdots, r)]_q\) locally repairable code
\begin{equation*}
    d \leq n - k - \left\lceil \frac{kt}{r} \right\rceil + t + 1. 
\end{equation*}
For codes with different localities, Bhadane and Thangaraj in \cite{8077091} established an upper bound on the minimum distance of \([n, k, d; (r_1, r_2, \cdots, r_t)]_q\) locally repairable code as
\begin{equation}\label{mul_bound_1}
    d \leq n - k + 1 - \sum_{i=1}^t \left\lfloor \frac{k-1}{\prod_{j=t+1-i}^t r_j} \right\rfloor,
\end{equation}
where \(r_1\le r_2\le\cdots\le r_t\). In \cite{Locally_Recoverable_Codes_From_Automorphism_Group_of_Function_Fields_of_Genus}, Bartoli et al. generalized the bound \eqref{min_d_bound_one} to the following upper bound of \([n, k, d; (r_1, r_2, \cdots, r_t)]_q\) locally repairable code
\begin{equation}\label{mul_bound_2}
    d \leq n - k - \left\lceil \frac{(k-1)t + 1}{1 + \sum_{i=1}^{t} r_i} \right\rceil + 2. 
\end{equation}
Throughout this paper, whenever this bound \eqref{mul_bound_2} is applicable, the bound \eqref{mul_bound_1} gives a lower value. Hence, we compare with \eqref{mul_bound_1} in what follows. We now define the Singleton-defect for an \([n, k, d; (r_1, r_2, \ldots, r_t)]_q\) locally repairable code \(\mathcal{C}\) as
\[
    \Delta(\mathcal{C}) = \frac{1}{n} \left( n - k - d + 1 - \sum_{i=1}^t \left\lfloor \frac{k-1}{\prod_{j=t+1-i}^t r_j} \right\rfloor \right), 
\]
where \(r_1\le r_2\le\cdots\le r_t\) and a smaller value of \(\Delta(\mathcal{C})\) corresponds to a better code. 


\subsection{Known Results}

Optimal locally repairable codes have attracted much attention. Constructions in \cite{6620541}, \cite{4276609}, \cite{Tamo20166661} suffer from alphabet size in code length. The construction in \cite{6284028} achieves alphabet size comparable to length, but only for specific length \(n = \lceil k/r \rceil (r+1)\) and the rate near \(1\). In \cite{A_Family_of_Optimal_Locally_Recoverable_Codes}, Tamo and Barg introduced a variation of Reed–Solomon codes to achieve local recoverability. These so-called locally repairable RS codes are optimal and have significantly lower locality than the RS codes themselves. However, their length is still shorter than the size of $\mathbb{F}_q$. A classical approach to obtain longer codes is to use algebraic curves with many rational points. In this way, Barg, Tamo, and Vlădut \cite{Locally_Recoverable_Codes_on_Algebraic_Curves} extended locally repairable RS codes to so-called locally repairable algebraic geometry codes, which in effect resulted in more locally repairable codes. The natural progression from these successes was the exploration of locally repairable codes using algebraic geometry methods, as indicated in \cite{Locally_Recoverable_Codes_from_Algebraic_Curves_and_Surfaces}, \cite{Locally_Recoverable_codes_with_availability_<i>t</i>≥2_from_fiber_products_of_curves}, \cite{Locally_Recoverable_codes_from_rational_maps}, \cite{Optimal_Locally_Repairable_Codes_Via_Elliptic_Curves}, \cite{Locally_Recoverable_J-affine_variety_codes}, \cite{Locally_Recoverable_codes_from_algebraic_curves_with_separated_variables},\cite{Construction_of_Optimal_Locally_Repairable_Codes_via_Automorphism_Groups_of_Rational_Function_Fields},\cite{Locally_Recoverable_Codes_From_Automorphism_Group_of_Function_Fields_of_Genus},\cite{MA2023105686}. Among these, a powerful technique is the construction of locally repairable codes from the automorphism groups of algebraic curves and algebraic function fields. Li et al. in~\cite{Optimal_Locally_Repairable_Codes_Via_Elliptic_Curves} provided a systematic method for producing \(q\)-ary optimal locally repairable codes with length \(O(q + 2\sqrt{q})\) via automorphism groups fixing the infinite place of elliptic curves and Jin et al. in~\cite{Construction_of_Optimal_Locally_Repairable_Codes_via_Automorphism_Groups_of_Rational_Function_Fields} constructed optimal locally repairable codes by employing automorphism groups of rational function fields. In~\cite{MA2023105686}, Ma and Xing overcame the locality constraint in~\cite{Optimal_Locally_Repairable_Codes_Via_Elliptic_Curves} by proposing a general framework for constructing optimal locally repairable codes via the automorphism groups fixing the infinite place and translation subgroups of elliptic curves. In~\cite{11223270}, Huang and Zhao constructed several \(q\)-ary either optimal or almost optimal locally repairable codes with length \(O(q + 4\sqrt{q})\) via automorphism groups of hyperelliptic curves by generalizing the methods in~\cite{Optimal_Locally_Repairable_Codes_Via_Elliptic_Curves} and~\cite{MA2023105686}. 

Locally repairable codes with availability \(t>1\) have also been the subject of extensive study in recent years. Numerous constructions of locally repairable codes with availability \(t\) have been proposed in the literature. For example, in~\cite{A_Family_of_Optimal_Locally_Recoverable_Codes}, Tamo and Barg proposed a polynomial-based construction of locally repairable codes with availability \(t\). In~\cite{7426763}, Hao et al. employed regular LDPC codes to construct binary locally repairable codes with \((r,t)\)-locality, i.e. each symbol has \(t\) disjoint recovering sets of size \(r\). In~\cite{han2024}, Han et al. utilized combinatorial designs to construct optimal locally repairable codes with \((4,t)\)-locality and dimension \(k\ge277\). In~\cite{Locally_Recoverable_codes_with_availability_<i>t</i>≥2_from_fiber_products_of_curves}, Haymaker et al. proposed a construction of locally repairable codes with availability \(t\ge2\) based on fiber products of algebraic curves. An effective class of constructions for locally repairable codes with availability also leverages the automorphism group structure of algebraic function fields. In~\cite{8865660}, Jin et al. extended the approach of \cite{Construction_of_Optimal_Locally_Repairable_Codes_via_Automorphism_Groups_of_Rational_Function_Fields} and employed the automorphism groups of rational function fields to construct locally repairable codes with availability \(t\). In \cite{Locally_Recoverable_Codes_From_Automorphism_Group_of_Function_Fields_of_Genus}, Bartoli et al. proposed a construction of locally repairable codes with $(r_1,\cdots,r_\delta)$-locality by working with a finite number of subgroups of cardinality $r_i +1$ of the automorphism group of a function field $F$ of genus $g \ge 1$ and applied their results to some well known families of function fields with many rational points. For asymptotic constructions of locally repairable codes with availability, Li et al. in~\cite{10619101} exploited the automorphism groups of function fields from the Garcia-Stichtenoth tower to obtain several families of asymptotically good such codes. Recently, in~\cite{MA2023105686}, Ma and Xing made a breakthrough by determining the precise structure of the automorphism group of elliptic function fields and analyzing its subgroup structure. Motivated by their result, we aim to construct locally repairable codes with availability \(t\) by utilizing the automorphism group of elliptic function fields in this paper.

\subsection{Main Results and Comparison}

In this paper, our first contribution is that we construct several new families of \(q\)-ary optimal (classical) locally repairable codes with length \(O(q+2\sqrt{q})\) and flexible locality by employing two special families of ordinary elliptic curves we introduce. Table~\ref{opt_parameter} presents a comparison between our optimal locally repairable codes and those constructed in previous researches using automorphism groups of algebraic curves and algebraic function fields. Compared with the construction of optimal locally repairable codes via elliptic function fields in \cite{Optimal_Locally_Repairable_Codes_Via_Elliptic_Curves} and \cite{MA2023105686}, our optimal locally repairable codes can be constructed over the finite field with characteristic \(p\) satisfying \( p \equiv 1 \pmod{12} \), a case for which no instances of optimal locally repairable codes were available in the previous results. Secondly, over the same finite field, the locality parameters of our optimal locally repairable codes also appear to differ from those in \cite{Optimal_Locally_Repairable_Codes_Via_Elliptic_Curves} and \cite{MA2023105686}, which implies that we complement the previous results. Our main results about optimal locally repairable codes are summarized as below. 

\begin{theorem}\label{opt_LRC_one_hA}
    Let \( p \) be a prime number and \(N\) be a positive integer. For any positive integer \( h \) with \(h\mid\mid N\), there exists an optimal \( p^2 \)-ary \([m(r+1), rt+1, (m-t)(r+1)]_{p^2}\) locally repairable code with locality \( r = h|A| - 1 \) for any integers \( t \) and \( m \) satisfying \(1 \leq t < m \leq \left\lceil \frac{N-2h}{r+1} \right\rceil-1 \) provided that \(|A|\), \( p \) and \(N\) satisfy one of the following cases:
    \begin{itemize}
        \item[(i)] \(N = p^2 + 2p\) and \(|A| = 2,3,6\) for \(p = 3u^2 + 3u + 1\) with\(u\in\mathbb{Z}\);
        \item[(ii)] \(N = p^2 + 2p - 3\) and \(|A| = 2,4\) for \( p \equiv 1 \pmod{4} \) and \(p = v^2 + 1\) with \(v\in\mathbb{Z}\).
    \end{itemize}
\end{theorem}

\begin{theorem}\label{opt_LRC_one_2h}
    Let \( p \) be a prime number and \(N\) be a positive integer. For any positive integer \( h \) with \(h\mid N\), there exists an optimal \( p^2 \)-ary \([m(r+1), rt+1, (m-t)(r+1)]_{p^2}\) locally repairable code with locality \( r = 2h - 1 \) for any integers \( t \) and \( m \) satisfying \(1 \leq t < m \leq \left\lceil \frac{N}{r+1} \right\rceil-2 \) provided that \( p \) and \(N\) satisfy one of the following cases:
    \begin{itemize}
        \item[(i)] \(N = p^2 + 2p\) for \(p = 3u^2 + 3u + 1\) with\(u\in\mathbb{Z}\);
        \item[(ii)] \(N = p^2 + 2p - 3\) for \( p \equiv 1 \pmod{4} \) and \(p = v^2 + 1\) with \(v\in\mathbb{Z}\).
    \end{itemize}
\end{theorem}

\begin{remark}
    For a prime power \(q\), there also exists an optimal \( q^2 \)-ary \([m(r+1), rt+1, (m-t)(r+1)]_{q^2}\) locally repairable code with locality \( r = h|A| - 1 \) or \( r = 2h - 1 \) provided that \(|A|\), \( q \) and \(N\) satisfy \(N = q^2 + 2q\) and \(|A| = 2,3,6\) for \(q = 3u^2 + 3u + 1\) with \(u\in\mathbb{Z}\). 
\end{remark}

    \begin{table*}[!t]
    \caption{Optimal (classical) locally repairable codes via automorphism groups of algebraic curves and algebraic function fields}
	\begin{center}
    \resizebox{1\textwidth}{!}{
	\begin{tabular}{cccccc}
        \toprule[1.2pt]
			Ref. & Length $n$ & Dimension $k$ & Locality $r$ & The range of $m$ and $t$ & Characteristic $p$ \\
        \midrule[1.2pt]
        
        Thm. \ref{opt_LRC_one_hA}(i) & $m(r+1)$ & $rt+1$ & \makecell{$h|A|-1$ and $h\mid\mid p^2 + 2p$ \\ $|A|=2,3,6$} & $1 \leq t < m \leq \left\lceil \frac{p^2 + 2p-2h}{r+1} \right\rceil-1$ & $p = 3u^2+3u+1$ \\

        \midrule[0.5pt]
        
        Thm. \ref{opt_LRC_one_hA}(ii) & $m(r+1)$ & $rt+1$ & \makecell{$h|A|-1$ and $h\mid\mid p^2 + 2p - 3$ \\ $|A|=2,4$} & $1 \leq t < m \leq \left\lceil \frac{p^2 + 2p - 3 -2h}{r+1} \right\rceil-1$ & \makecell{\( p \equiv 1 \pmod{4} \) \\ and \(p = v^2 + 1\)} \\

        \midrule[0.5pt]
        
        Thm. \ref{opt_LRC_one_2h}(i) & $m(r+1)$ & $rt+1$ & $2h-1$\text{ and }$h\mid p^2 + 2p$ & $1 \leq t < m \leq \left\lceil \frac{p^2 + 2p}{r+1} \right\rceil-2$ & $p = 3u^2+3u+1$ \\

        \midrule[0.5pt]
        
        Thm. \ref{opt_LRC_one_2h}(ii) & $m(r+1)$ & $rt+1$ & $2h-1$\text{ and }$h\mid p^2 + 2p - 3$ & $1 \leq t < m \leq \left\lceil \frac{p^2 + 2p - 3}{r+1} \right\rceil-2$ & \makecell{\( p \equiv 1 \pmod{4} \) \\ and \(p = v^2 + 1\)}\\ 

        \midrule[0.5pt]

         \cite[Thm. 1]{Optimal_Locally_Repairable_Codes_Via_Elliptic_Curves} & $3m$ & $2t+1$ & $2$ & $0\le t<m\le\left\lfloor\frac{q+2\sqrt{q}}{3}\right\rfloor$ & \makecell{\(p = 3\) \\ \(p \equiv 2 \pmod{3}\)} \\
        \midrule[0.5pt]
        
          \cite[Thm. 2]{Optimal_Locally_Repairable_Codes_Via_Elliptic_Curves} & $m(r+1)$ & $rt-(r-1)$ & $3,5,7,11,23$ & $1\le t\le m\le\left\lfloor\frac{q+2\sqrt{q}-r-2}{r+1}\right\rfloor$ & \makecell{\(p = 2, 3\) \\ \( p \equiv 3 \pmod{4} \) \\ \(p \equiv 2 \pmod{3}\)} \\
        \midrule[0.5pt]
        
          \cite[Thm. V.3]{Construction_of_Optimal_Locally_Repairable_Codes_via_Automorphism_Groups_of_Rational_Function_Fields} & $m(r+1)$ & $rt$ & $(r+1)\mid(q+1)$ & $1\le t\le m\le\frac{q+1}{r+1}$ & Arbitrary \\
        \midrule[0.5pt]
        
           \cite[Sec. IV]{Construction_of_Optimal_Locally_Repairable_Codes_via_Automorphism_Groups_of_Rational_Function_Fields} & $m(r+1)$ & $rt$ & \makecell{$p^v-1,up^v-1$ \\ or $(r+1)\mid (q-1)$} & \makecell{$1\le t\le m\le L$ \\ $L = \frac{q}{r+1},\frac{q-1}{r+1}\text{ or }\frac{q-p^v}{up^v}$ } & Arbitrary \\
        \midrule[0.5pt]

        \cite[Thm. 1.1]{MA2023105686} & $m(r+1)$ & $rt+1$ & $2h-1\text{ and }h\mid (\sqrt{q}+1)^2$ & $1\le t<m\le\left\lceil\frac{q+2\sqrt{q}-2r-1}{r+1}\right\rceil$ & \makecell{\(p = 2, 3\) \\ \( p \equiv 3 \pmod{4} \) \\ \(p \equiv 2 \pmod{3}\)} \\
        \midrule[0.5pt]

        \cite[Thm. 1.2]{MA2023105686} & $m(r+1)$ & $rt+1$ & \makecell{$h^2|A|-1$ and $h\mid \sqrt{q}+1$ \\ $|A|=2,3,4,6,8,12,24$} & $1\le t<m\le\left\lceil\frac{q+2\sqrt{q}-2h^2-r}{r+1}\right\rceil$ & \makecell{\(p = 2, 3\) \\ \( p \equiv 3 \pmod{4} \) \\ \(p \equiv 2 \pmod{3}\)} \\
        \midrule[0.5pt]
        
        \cite[Thm. 3(i)]{11223270} & $4m$ & $3m-2$ & $3$ & $1\le m\le\left\lfloor\frac{q+4\sqrt{q}-9}{4}\right\rfloor$ & $p\equiv 5$ or $7 \pmod{8}$\\
        \midrule[0.5pt]
        
        \cite[Thm. 3(ii)]{11223270} & $3m$ & $2m-1$ & $2$ & $1\le m\le\left\lfloor\frac{N(F)}{3}\right\rfloor-1$ & $p\neq3$ and $3\mid q-1$\\

        \midrule[0.5pt]
        
        \cite[Thm. 3(iii)]{11223270} & $5m$ & $4m-2$ & $4$ & $1\le  m\le\big\lfloor\frac{q+4\sqrt{q}-11}{5}\big\rfloor$ & \makecell{\(p = 5\) \\ \( p \equiv -1 \pmod{5} \)}\\
        \bottomrule[1.2pt]
	\end{tabular}
    }
    \end{center}
	\label{opt_parameter}
    \end{table*}

Another important contribution of this paper is that we introduce a general framework for constructing locally repairable codes with two recovering sets (i.e. availability \(2\)) via automorphism groups of elliptic function fields. Following the construction framework in \cite{MA2023105686}, we find that it seems unable to obtain locally repairable codes with two recovering sets and flexible parameters via elliptic function fields. To overcome this limitation, we propose a novel construction for determining the functions \(e_i\) in the construction of locally repairable codes. This ensures that the function spaces corresponding to the two distinct recovering sets lie within a common Riemann–Roch space, thereby guaranteeing that their intersection is sufficiently large. Using this framework and utilizing maximal elliptic function fields and specific elliptic curves introduced by us, we construct several distinct families of \(q\)-ary locally repairable codes with two recovering sets and length \(O(q+2\sqrt{q})\). Our main results about locally repairable codes with two recovering sets are summarized as below.

\begin{theorem}\label{lrc_two_max}
    Let \( q = p^a \) for any prime \( p \) and any even integer \( a > 0 \). Let \(\sqrt{q}+1 = \prod_\ell \ell^{h_\ell}\)  be the prime factorization of \( \sqrt{q} + 1 \). For any integer \( d_0 \) with \( 1 \leq d_0 < n \) and any positive divisor \( h \) of \( \sqrt{q} + 1 \) with
    \begin{equation*}
        h^2|A_1||A_2|>\prod_\ell \ell^{2\min\{2\nu_\ell(h),h_\ell\}}-h^2,
    \end{equation*}
    there exists a \( q \)-ary \([n = mh^2|A_1||A_2|, k, d;(r_1,r_2)]_q\) locally repairable code where \( r_1 = h^2|A_1| - 1 \), \(r_2 = h^2(|A_2|-1)\), \(d \ge n - L \ge d_0\), 
    \[
        k \geq \bigg\lfloor\frac{n-d_0}{r_1 + 1}\bigg\rfloor r_1 + \bigg\lfloor\frac{n-d_0}{r_2 + h^2}\bigg\rfloor r_2  + 2 - L,
    \]
    and
    \[
        L = \max\left\{\left\lfloor\frac{n-d_0}{r_1 + 1}\right\rfloor (r_1+1),\left\lfloor\frac{n-d_0}{r_2 + h^2}\right\rfloor (r_2 + h^2)\right\}, 
    \]
    for any integer \( m \) satisfying \(1 \leq m \leq \left\lceil \frac{q+2\sqrt{q}+1-2h^2}{h^2|A_1||A_2|} \right\rceil - 1 \), provided that \(|A_1|\), \(|A_2|\) and \( p \) satisfy one of the following cases:
     
    \begin{itemize}
        \item[(i)] \(|A_1| = 2,4,8\) and \(|A_2| = 3\) for \( p = 2 \);  
        \item[(ii)] \(|A_1| = 2,4\) and \(|A_2| = 3\) for \( p = 3 \);  
        \item[(iii)] \(|A_1| = 2\) and \(|A_2| = 3\) for \( p \equiv 2 \pmod{3} \) and \( p \neq 2 \).
    \end{itemize}
    (\(\nu_\ell(\cdot)\) is the normalized valuation of \(\mathbb{Z}\) corresponding to the prime number \(\ell\).)
\end{theorem}

\begin{theorem}\label{lrc_two_nonmax}
    Let \( q \) be a prime power. For any integer \( d_0 \) with \( 1 \leq d_0 < n \) and any positive divisor \( h \) with \(h\mid\mid q^2+2q\), there exists a \( q^2 \)-ary \([n = mh|A_1||A_2|, k, d;(r_1,r_2)]_{q^2}\) locally repairable code where \( r_1 = h|A_1| - 1 \), \(r_2 = h(|A_2|-1)\), \(d \ge n - L \ge d_0\), 
    \[
        k \geq \bigg\lfloor\frac{n-d_0}{r_1 + 1}\bigg\rfloor r_1 + \bigg\lfloor\frac{n-d_0}{r_2 + h}\bigg\rfloor r_2  + 2 - L,
    \]
    and
    \[
        L = \max\left\{\left\lfloor\frac{n-d_0}{r_1 + 1}\right\rfloor (r_1+1),\left\lfloor\frac{n-d_0}{r_2 + h}\right\rfloor (r_2 + h)\right\}, 
    \]
    for any integer \( m \) satisfying \(1 \leq m \leq \left\lceil \frac{q^2+2q-2h}{h|A_1||A_2|} \right\rceil - 1 \), provided that \(|A_1|\), \(|A_2|\) and \( q \) satisfy that \(|A_1| = 2\) and \(|A_2| = 3\) for \(q = 3u^2 + 3u + 1\) with \(u\in\mathbb{Z}\). 
\end{theorem}

\begin{remark}
    \begin{itemize}
        \item[(i)] The parameters \(|A_1|\) and \(|A_2|\) in Theorems~\ref{lrc_two_max} and \ref{lrc_two_nonmax} are interchangeable. 
        \item[(ii)] We can provide a special construction of \(q\)-ary \([n, k, d;(5,6)]_q\) and \([n, k, d;(8,3)]_q\) locally repairable codes for any \( q = 4^{2a+1} \) with \( a \in \mathbb{N} \).
    \end{itemize}
\end{remark}

\subsection{Organization}

The paper is organized as follows. In Section~\ref{pre}, we provide some preliminaries about algebraic function fields, algebraic geometry codes, elliptic curves and elliptic function fields, automorphism groups of elliptic curves and elliptic function fields. In Section~\ref{construct}, we will divide this into two parts, and begin with an alternative method for determining the functions \(e_i\) in the construction of locally repairable codes. In the second part, we present several new families of optimal locally repairable codes. In Section~\ref{construct_two}, we first provide a general framework for constructing locally repairable codes with two recovering sets via automorphism groups of elliptic function fields. And then, through this framework, we construct several distinct families of locally repairable codes equipped with two recovering sets. In Section~\ref{conclusion}, we give a summary of the entire article.

\section{Preliminaries}\label{pre}
In this section, we provide a concise review of some preliminaries related to algebraic function fields, algebraic geometry codes, elliptic curves and elliptic function fields, automorphism groups of elliptic curves and elliptic function fields.

\subsection{Algebraic Function Fields}
Let $F/\mathbb{F}_q$ be a function field of genus $g(F)$ with the full constant field $\mathbb{F}_q$. Let $\mathbb{P}_F$ denote the set of places of $F$. The discrete valuation associated with $P\in\mathbb{P}_F$ is denoted by $v_P$. The divisor group of $F/\mathbb{F}_q$ is defined as the free abelian group which is generated by $\mathbb{P}_F$; it is denoted by $\operatorname{Div}(F)$. Assume that $D=\sum_{P\in\mathbb{P}_F}n_PP$ is a divisor such that almost all $n_P=0$. Then the degree of $D$ is defined by $\deg D=\sum_{P\in\mathbb{P}_F}n_P\deg P$ and the support of $D$ by $\operatorname{supp}(D)=\{P\in \mathbb{P}_F:n_P\neq0\}$. For $z\in F$, the principal divisor of $z$ is denoted by $(z)^F$, the zero divisor of $z$ is denoted by $(z)^F_0$ and the pole divisor of $z$ is denoted by $(z)^F_\infty$. Two divisors $D,D^\prime\in\operatorname{Div}(F)$ are said to be linearly equivalent, written $D\sim D^\prime$, if $D=D^\prime+(z)^F$ for some $z\in F\backslash\{0\}$. 

The Riemann-Roch space associated to the divisor $D\in\operatorname{Div}(F)$ is the finite-dimensional $\mathbb{F}_q$-vector space
\[
    \mathcal{L}(D)=\big\{z\in F: (z)^F+D\ge 0 \big\}\cup\big\{0\big\}.
\]
The dimension of $\mathcal{L}(D)$ is given by $\ell(D)$ and it satisfies the Riemann-Roch Theorem \cite[Thm. 1.5.15]{Algebraic_Function_Fields_and_Codes}. Let $\operatorname{Aut}(F/\mathbb{F}_q)$ be the automorphism group of $F$ over $\mathbb{F}_q$, i.e.
\[
    \operatorname{Aut}(F/\mathbb{F}_q)=\big\{\sigma:\sigma \text{ is an } \mathbb{F}_q\text{-automorphism of } F\big\}.
\]
Now let $\mathcal{G}$ be a finite subgroup of $\operatorname{Aut}(F/\mathbb{F}_q)$. The fixed subfield of $F$ with respect to $\mathcal{G}$ is defined by
\[
    F^{\mathcal{G}}=\big\{z\in F:\sigma(z)=z \text{ for all } \sigma\in\mathcal{G}\big\}.
\]
From the Galois theory, $F/F^{\mathcal{G}}$ is a Galois extension with $\operatorname{Gal}(F/F^{\mathcal{G}})=\mathcal{G}$. Moreover, $F^{\mathcal{G}}/\mathbb{F}_q$ is also a function field with the full constant field $\mathbb{F}_q$. By \cite[Lem. 3.5.2]{Algebraic_Function_Fields_and_Codes}, for any automorphism $\sigma\in\operatorname{Gal}(F/F^{\mathcal{G}})$ and any place $P\in\mathbb{P}_F$, then $\sigma(P):=\{\sigma(z):z\in P\}$ is a place of $F$ as well. Let $g(F^{\mathcal{G}})$ denote the genus of $F^{\mathcal{G}}$. Then the Hurwitz Genus Formula\cite[Thm. 3.4.13]{Algebraic_Function_Fields_and_Codes} yields
\begin{equation*}
    2g(F)-2=[F:F^{\mathcal{G}}](2g(F^{\mathcal{G}})-2)+\deg {\rm{Diff}}(F/F^{\mathcal{G}}),
\end{equation*}
where ${\rm{Diff}}(F/F^{\mathcal{G}})$ stands for the different of $F/F^{\mathcal{G}}$.

\subsection{Algebraic Geometry Codes}
For more details of algebraic geometry codes, the reader may refer to \cite{Algebraic-Geometric_Codes}. Let $F/\mathbb{F}_q$ be a function field of genus $g$ with the full constant field $\mathbb{F}_q$. Let $P_1,\cdots,P_n\in\mathbb{P}_F$ be $n$ pairwise distinct rational places of $F$ and $D=\sum_{i=1}^nP_i$. For a divisor $G$ of $F/\mathbb{F}_q$ with $2g-2<\deg G<n$ and ${\rm{supp}}(G)\cap {\rm{supp}}(D)=\emptyset$, the algebraic geometry code associated with the divisors $D$ and $G$ is defined as
\[
    C_{\mathcal{L}}(D,G)=\big\{(x(P_1),\,\cdots,\,x(P_n)):x\in\mathcal{L}(G)\big\}\subseteq\mathbb{F}_q^n,
\]
where $\mathcal{L}(G)$ is the Riemann-Roch space with the dimension $\dim_{\mathbb{F}_q}\mathcal{L}(G)=\deg G+1-g$ from the Riemann-Roch Theorem \cite[Thm. 1.5.15]{Algebraic_Function_Fields_and_Codes}. Then the code $C_{\mathcal{L}}(D,G)$ is an $[n,k,d]_q$ linear code with dimension $k=\dim_{\mathbb{F}_q}\mathcal{L}(G)$ and minimum distance $d\ge n-\deg G$. If $V$ is a subspace of $\mathcal{L}(G)$, then we can define a subcode of $C_{\mathcal{L}}(D,G)$ by 
\[
    C_{\mathcal{L}}(D,V)=\big\{(x(P_1),\,\cdots,\,x(P_n)):x\in V\big\}.
\]
Then the dimension of $C_{\mathcal{L}}(D,V)$ is the dimension of the space $V$ over $\mathbb{F}_q$ and the minimum distance of $C_{\mathcal{L}}(D,V)$ is still lower bounded by $n-\deg G$.

\subsection{Elliptic Curves and Elliptic Function Fields}

Let $q$ be a power of an odd prime $p$. An elliptic curve \( \mathfrak{E} \) defined over a finite field \( \mathbb{F}_q \) can be given by a nonsingular Weierstrass equation
\begin{equation}\label{weierstrass}
    y^2 + a_1xy + a_3y = x^3 + a_2x^2 + a_4x + a_6,
\end{equation}
where \( a_i \) are elements of \( \mathbb{F}_q \). We write $\mathfrak{E}/\mathbb{F}_q$ as an elliptic curve $\mathfrak{E}$ defined over $\mathbb{F}_q$. Denote by $E/\mathbb{F}_q=\mathbb{F}_q(\mathfrak{E})=\mathbb{F}_q(x,y)$ and $\mathfrak{E}(\mathbb{F}_q)$ the elliptic function field of $\mathfrak{E}/\mathbb{F}_q$ over \(\mathbb{F}_q\) and the set of $\mathbb{F}_q$-rational points, respectively. The genus of \( E \) is \( g(E) = 1 \). Let $\mathbb{P}_E$ be the set of all places of $E$ and $\mathbb{P}_E^1=\{P\in\mathbb{P}_E:\deg P =1\}$ be the set of rational places of $E$. There is a one-to-one correspondence between $\mathfrak{E}(\mathbb{F}_q)$ and $\mathbb{P}_E^1$. More specifically, the rational point $(\alpha,\beta)$ on $\mathfrak{E}$ corresponds to the unique common zero of $x-\alpha$ and $y-\beta$, denoted by $P_{\alpha,\beta}$; and the unique infinite place \( \mathcal{O} \) is the common pole of \( x \) and \( y \).  

Let \(N(E) = |\mathbb{P}_E^1|\) i.e. the number of the rational places of \(E/\mathbb{F}_q\). The following lemma gives a bound on the size of $N(E)$, which is called the Hasse-Weil Bound \cite[Thm. 5.2.3]{Algebraic_Function_Fields_and_Codes}.

\begin{lemma}
    Let $E/\mathbb{F}_q$ be an elliptic function field with $N(E)$ rational places. Then we have the following Hasse-Weil Bound
    \begin{equation}\label{Hasse-Weil_Bound}
        |N(E)-q-1|\le 2\sqrt{q}.
    \end{equation}
\end{lemma}

If the number $N(E)$ attains the upper bound \eqref{Hasse-Weil_Bound}, then $\mathfrak{E}/\mathbb{F}_q$ is called a maximal elliptic curve and $E/\mathbb{F}_q$ is called a maximal function field. The divisor group of $E/\mathbb{F}_q$ is defined as the free abelian group which is generated by $\mathbb{P}_E$; it is denoted by $\operatorname{Div}(E)$. Two divisors $D,D^\prime\in\operatorname{Div}(E)$ are said to be linearly equivalent, written $D\sim D^\prime$, if $D=D^\prime+(z)^{E}$ for some $z\in E^*$. The group of principal divisors of $E$ is $\operatorname{Princ}(E) = \{(z)^E:z\in E^*\}$. We denote the divisor class group of $\mathfrak{E}$ by $\operatorname{Cl}(E)$, which is defined as
\[
    \operatorname{Cl}(E)=\operatorname{Div}(E)/\operatorname{Princ}(E).
\]
The class of a divisor $D$ in $\operatorname{Cl}(E)$ will be denoted by $[D]$. We define $\operatorname{Div}^0(E)$ as the degree zero subgroup of $\operatorname{Div}(E)$ and define $\operatorname{Cl}^0(E)$ as the degree zero subgroup of $\operatorname{Cl}(E)$. From \cite[Prop. 6.1.7]{Algebraic_Function_Fields_and_Codes}, there is a group isomorphism between \( \mathbb{P}_E^1 \) and \( \operatorname{Cl}^0(E) \) given by
\[
    \Phi : 
    \begin{cases} 
        \mathbb{P}_E^1 \to \operatorname{Cl}^0(E), \\ 
        P \mapsto [P - \mathcal{O}]. 
    \end{cases}
\]
The group operation of \( \mathbb{P}_E^1 \) is defined by \( P \oplus Q = R \iff P + Q \sim R + \mathcal{O} \) for any \( P, Q \in \mathbb{P}_E^1 \). Denote by \([2]P = P \oplus P\) and define \([m+1]P = [m]P \oplus P\) recursively. In fact, \( \mathbb{P}_E^1 \) is an abelian group and the place \( \mathcal{O} \) is the zero element of the group \( \mathbb{P}_E^1 \). Let \(E[n]\) be the subgroup of \(\mathbb{P}_E^1\) consisting of those places \(P\) annihilated by \(n\). The following result is an important conclusion about the rational places of the elliptic function field \cite[Prop. 6.1.6]{Algebraic_Function_Fields_and_Codes}. 

\begin{lemma}\label{P=Q}
    Let \( E/\mathbb{F}_q \) be an elliptic function field and let \( P, Q \) be two rational places of \( E \). Then we have  
    \[
        P \sim Q \text{ if and only if } P = Q.
    \]
\end{lemma}

From \cite{silvermanArithmeticEllipticCurves2009}, two elliptic curves $\mathfrak{E}_1$ and $\mathfrak{E}_2$ over $\mathbb{F}_q$ are called isogenous if there exists a non-constant smooth $\mathbb{F}_q$-morphism from $\mathfrak{E}_1$ to $\mathfrak{E}_2$ mapping the zero element of $\mathfrak{E}_1$ to that of $\mathfrak{E}_2$. A fundamental result in the theory of elliptic curves states that two elliptic curves defined over $\mathbb{F}_q$ are isogenous if and only if they have the same number of $\mathbb{F}_q$-rational points. More precisely, we have the following lemma\cite{waterhouse_abelian_1969}.

\begin{lemma}\label{iso}
    The isogeny classes of elliptic curves over \( \mathbb{F}_q \) for \( q = p^a \) are in one-to-one correspondence with the rational integers \( t \) having \( |t| \leq 2\sqrt{q} \) and satisfying one of the following conditions:
    \begin{itemize}
        \item[(i)] \( (t, p) = 1 \);
        \item[(ii)] If \( a \) is even: \( t = \pm 2\sqrt{q} \);
        \item[(iii)] If \( a \) is even and \( p \not\equiv 1 \pmod{3} \): \( t = \pm \sqrt{q} \);
        \item[(iv)] If \( a \) is odd and \( p = 2 \) or \( 3 \): \( t = \pm p^{\frac{a+1}{2}} \);
        \item[(v)] If either (i) \( a \) is odd or (ii) \( a \) is even and \( p \not\equiv 1 \pmod{4} \): \( t = 0 \).
    \end{itemize}
    The first of these is ordinary, the rest are supersingular. Furthermore, an elliptic curve in the isogeny class corresponding to \( t \) has \( q+1-t \) rational points.
\end{lemma}

The abelian group structure of $\mathbb{P}_E^1$ follows from \cite[Thm. 3]{Ruck1987301} and \cite[Thm. 9.97]{hirschfeld2008algebraic}.

\begin{lemma}
    Let \( \mathbb{F}_q \) be the finite field with \( q = p^s \) elements. Let \( h = \prod_{\ell} \ell^{h_\ell} \) be a possible number of rational places of an elliptic curve \( E \) over \( \mathbb{F}_q \). Then all the possible groups \( \mathbb{P}^1_E \) are the following
    \[
        \mathbb{Z}/p^{h_p} \mathbb{Z} \times \prod_{\ell \neq p} (\mathbb{Z}/\ell^{a_\ell} \mathbb{Z} \times \mathbb{Z}/\ell^{h_\ell - a_\ell} \mathbb{Z})
    \]
    with
    \begin{itemize}
        \item[(I)] In case (ii) of Lemma \ref{iso}: Each \( a_\ell \) is equal to \( h_\ell / 2 \), i.e., \( \mathbb{P}^1_E \cong \mathbb{Z}/(\sqrt{q} \pm 1)\mathbb{Z} \times \mathbb{Z}/(\sqrt{q} \pm 1)\mathbb{Z} \).
        \item[(II)] In other cases of Lemma \ref{iso}: \( a_\ell \) is an arbitrary integer satisfying \( 0 \leq a_\ell \leq \min\{\nu_\ell(q-1), \left\lceil h_\ell/2\right\rceil\} \). In cases (iii) and (iv) of Lemma \ref{iso}: \( \mathbb{P}^1_E \cong \mathbb{Z}/h\mathbb{Z} \). In cases (v) of Lemma \ref{iso}: if \( q \not\equiv -1 \pmod{4} \), then \( \mathbb{P}^1_E \cong \mathbb{Z}/(q+1)\mathbb{Z} \); otherwise, \( \mathbb{P}^1_E \cong \mathbb{Z}/(q+1)\mathbb{Z} \) or \( \mathbb{P}^1_E \cong \mathbb{Z}/2\mathbb{Z} \times \mathbb{Z}/\frac{q+1}{2}\mathbb{Z} \).
    \end{itemize}
\end{lemma}

\subsection{Automorphism Groups of Elliptic Curves and Elliptic Function Fields}

We begin by reviewing the automorphism groups of elliptic curves. Let $\mathfrak{E}/\mathbb{F}_q$ be an elliptic curve defined by the Weierstrass equation \eqref{weierstrass}, and denote by $\operatorname{Aut}(\mathfrak{E})$ its group of automorphisms over the algebraic closure $\overline{\mathbb{F}}_q$. Every such automorphism must fix the point at infinity $\mathcal{O}$. The following lemma is a full description of the automorphism group which is given in \cite[Thm. 3.10.1]{silvermanArithmeticEllipticCurves2009}.

\begin{lemma}
    Let \( \mathfrak{E}/\mathbb{F}_q \) be an elliptic curve with \( j \)-invariant \( j(\mathfrak{E}) \). Then the order of \( \operatorname{Aut}(\mathfrak{E}) \) divides 24. More precisely speaking, the order of \( \operatorname{Aut}(\mathfrak{E}) \) is given by the following list:
    \begin{itemize}
        \item[(i)] \( |\operatorname{Aut}(\mathfrak{E})| = 2 \) if \( j(\mathfrak{E}) \neq 0, 1728 \);
        \item[(ii)] \( |\operatorname{Aut}(\mathfrak{E})| = 4 \) if \( j(\mathfrak{E}) = 1728 \) and \( \text{char}(\mathbb{F}_q) \neq 2, 3 \);
        \item[(iii)] \( |\operatorname{Aut}(\mathfrak{E})| = 6 \) if \( j(\mathfrak{E}) = 0 \) and \( \text{char}(\mathbb{F}_q) \neq 2, 3 \);
        \item[(iv)] \( |\operatorname{Aut}(\mathfrak{E})| = 12 \) if \( j(\mathfrak{E}) = 0 = 1728 \) and \( \text{char}(\mathbb{F}_q) = 3 \);
        \item[(v)] \( |\operatorname{Aut}(\mathfrak{E})| = 24 \) if \( j(\mathfrak{E}) = 0 = 1728 \) and \( \text{char}(\mathbb{F}_q) = 2 \).
    \end{itemize}
\end{lemma}

Now turn to the corresponding function field $E/\mathbb{F}_q$ of $\mathfrak{E}/\mathbb{F}_q$ and let $\operatorname{Aut}(E/\mathbb{F}_q)$ be the group of its $\mathbb{F}_q$-automorphisms. For any $\sigma \in \operatorname{Aut}(E/\mathbb{F}_q)$ and any place $P \in \mathbb{P}_E$, we know that $\sigma(P)$ is also a place of $E$\cite[Lem. 3.5.2]{Algebraic_Function_Fields_and_Codes}. We define $\operatorname{Aut}(E,\mathcal{O})$ as the subgroup of $\operatorname{Aut}(E/\mathbb{F}_q)$ fixing the infinite place $\mathcal{O}$. This subgroup represents the intersection of the curve's automorphism group with the function field's $\mathbb{F}_q$-automorphisms, i.e. $\operatorname{Aut}(E,\mathcal{O}) = \operatorname{Aut}(\mathfrak{E}) \cap \operatorname{Aut}(E/\mathbb{F}_q)$. Furthermore, for each rational place $Q \in \mathbb{P}_E^1$, the translation-by-$Q$ map $\tau_Q$ defined by $\tau_Q(P) = P \oplus Q$ induces an $\mathbb{F}_q$-automorphism. The collection of all such maps, i.e. $T_E = \{ \tau_Q:Q\in \mathbb{P}_E^1 \}$, is called the translation group and is isomorphic to $\mathbb{P}_E^1$. The following lemmas \cite[Thm. 3.1 and Prop. 3.2]{MA2023105686} characterize the full automorphism group of an elliptic function field and its subgroup structures.

\begin{lemma}\label{ell_func_auto}
    Let \( E/\mathbb{F}_q \) be an elliptic function field. The automorphism group of elliptic function field \( E \) over \( \mathbb{F}_q \) is the semidirect product of the translation group \( T_E \) and the stabilizer \( \operatorname{Aut}(E, \mathcal{O}) \) of the infinite place \( \mathcal{O} \), i.e., 
    \[
        \operatorname{Aut}(E/\mathbb{F}_q) = T_E \rtimes \operatorname{Aut}(E, \mathcal{O}).
    \]
    The group law of \( \operatorname{Aut}(E/\mathbb{F}_q) \) is given by \( (\tau_P\alpha) \cdot (\tau_Q\beta) = \tau_{P\oplus\alpha(Q)} \cdot \alpha\beta \) for any \( \tau_P, \tau_Q \in T_E \) and \( \alpha, \beta \in \operatorname{Aut}(E, \mathcal{O}) \).
\end{lemma}

\begin{lemma}\label{ell_func_auto_sub}
    Let \( E/\mathbb{F}_q \) be an elliptic function field and let \( G \) be a subgroup of  \( \operatorname{Aut}(E/\mathbb{F}_q) \). Then, we have \( G \cong (T_E \cap G) \rtimes \pi(G) \), i.e., every subgroup of \( \operatorname{Aut}(E/\mathbb{F}_q) \) is isomorphic to a semidirect product of a subgroup of \( T_E \) and a subgroup of \( \operatorname{Aut}(E, \mathcal{O}) \).
\end{lemma}

Lemma \ref{ell_func_auto_sub} states that for any subgroup $G$ of $\operatorname{Aut}(E/\mathbb{F}_q)$, there exist a subgroup $T$ of $T_E$ and a subgroup $A$ of $\operatorname{Aut}(E, \mathcal{O})$ such that $G \cong T \rtimes A = TA$. Conversely, given a subgroup $T$ of $T_E$ and a subgroup $A$ of $\operatorname{Aut}(E, \mathcal{O})$, the product $TA$ is a subgroup of $\operatorname{Aut}(E/\mathbb{F}_q)$ if and only if a specific condition holds. This condition is provided by the following lemma\cite[Prop. 3.3]{MA2023105686}.

\begin{lemma}\label{sub}
    Let \( T \) be a subgroup of the translation group \( T_E \) and let \( A \) be a subgroup of \( \operatorname{Aut}(E, \mathcal{O}) \). Then \( TA \) is a subgroup of \( \operatorname{Aut}(E / \mathbb{F}_q) \) if and only if \( \tau_{\sigma^{-1}(Q)} \in T \) for all \( \sigma \in A \) and \( \tau_Q \in T \).
\end{lemma}

For an elliptic function field $E/\mathbb{F}_q$ and a subgroup $G$ of its automorphism group $\operatorname{Aut}(E/\mathbb{F}_q)$, the Galois extension $E/E^G$ is ramified at only a few places, as shown by the following lemma\cite[Prop. 4.1]{MA2023105686}.

\begin{lemma}\label{split_place}
    Let \( E/\mathbb{F}_q \) be an elliptic function field with \( N(E) \) rational places. Let \( T \) be a subgroup of the translation group \( T_E \) and let \( A \) be a nontrivial subgroup of \( \operatorname{Aut}(E, \mathcal{O}) \) such that \( G = TA \) is a subgroup of \( \operatorname{Aut}(E/\mathbb{F}_q) \). Let \( |G| = r + 1 \) and \( F = E^G \). Then there are at most \( r + 1 + 2|T| \) rational places of \( E \) that are ramified in \( E/F \). All unramified rational places of \( E \) split completely in \( E/F \).
\end{lemma}

In \cite{Optimal_Locally_Repairable_Codes_Via_Elliptic_Curves} and \cite{MA2023105686}, the authors investigated maximal elliptic function fields and subgroups of the automorphism group $\operatorname{Aut}(E / \mathbb{F}_q)$ satisfying Lemma \ref{sub}, where the subgroup \(A\) of \( \operatorname{Aut}(E, \mathcal{O}) \) is nontrivial, and pioneeringly constructed several families of locally repairable codes with flexible localities via maximal elliptic function fields. In this paper, aiming to obtain locally repairable codes with a wider variety of localities and long lengths, we now study certain elliptic curves which approach maximal and their automorphism groups \( \operatorname{Aut}(E, \mathcal{O}) \). Many of the things we talk about below are known to experts. However, we have been unable to found any one work that contains all of these facts and their proofs. Therefore, we prove the facts that we need, as well as some related extensions here. 

\begin{proposition}\label{new_ellp_cuv}
    Let \( p \) be a  prime number.
    \begin{itemize}
        \item[(i)] Let \( \mathfrak{E}/\mathbb{F}_{p^2} \) be an elliptic curve over \( \mathbb{F}_{p^2} \) defined by the equation \( y^2 = x^3 + b \) for some \( b \in \mathbb{F}_{p}^* \) and \( E/\mathbb{F}_{p^2} \) be the elliptic function field of \( \mathfrak{E}/\mathbb{F}_{p^2} \). Then \(N(E) = {p^2} + 2p\) if and only if \(p = 3u^2 + 3u + 1\) with \(u\in\mathbb{Z}\). In this case, we have \(|\operatorname{Aut}(E, \mathcal{O})| = 6\). 
        \item[(ii)] Let \( \mathfrak{E}/\mathbb{F}_{p^2} \) be an elliptic curve over \( \mathbb{F}_{p^2} \) defined by the equation \( y^2 = x^3 + ax \) for some \( a \in \mathbb{F}_{p}^* \) and \( E/\mathbb{F}_{p^2} \) be the elliptic function field of \( \mathfrak{E}/\mathbb{F}_{p^2} \). Then \(N(E) = {p^2} + 2p - 3\) if and only if \( p \equiv 1 \pmod{4} \) and \(p = v^2 + 1\) with \(v\in\mathbb{Z}\). In this case, we have \(|\operatorname{Aut}(E, \mathcal{O})| = 4\). 
    \end{itemize}
\end{proposition}
\begin{proof}
    By Lemma \ref{iso}, we know that the elliptic curves under consideration are ordinary. Moreover, by \cite[Prop. 4.37]{washington_elliptic_2008}, the elliptic curve \( y^2 = x^3 + b \) over \( \mathbb{F}_{p^2} \) is ordinary if and only if \( p \equiv 1 \pmod{3} \), and the elliptic curve \( y^2 = x^3 + ax \) over \( \mathbb{F}_{p^2} \) is ordinary if and only if \( p \equiv 1 \pmod{4} \). Let \(|\mathfrak{E}(\mathbb{F}_p)| = p + 1 - t = p + 1 - (\pi + \bar{\pi})\) where \(\pi\) is the Frobenius endomorphism. Then we have 
    \[
        |\mathfrak{E}(\mathbb{F}_{p^2})| = N(E) = (p + 1)^2 - t^2 = p^2 + 2p + 1 - t^2.
    \]
    
    (i) In this case, the $j$-invariant \(j(\mathfrak{E})\) of \( \mathfrak{E} \) is equal to 0 and the endomorphism ring of \( \mathfrak{E} \) is \({\rm End}(\mathfrak{E}) = \mathbb{Z}[\omega]\) where \(\omega = \frac{-1 + \sqrt{-3}}{2}\). Let \( \pi = u + v\omega\) with \( u,v\in\mathbb{Z} \). Then
    \[
        N(E) = p^2 + 2p \Longleftrightarrow t = \pi + \bar{\pi} = \pm 1.
    \]
    From the proposition of the Frobenius endomorphism, we have \( \pi\bar{\pi} = p \). Hence, after straightforward calculations, we conclude that
    \[
        N(E) = p^2 + 2p \Longleftrightarrow \begin{cases}
            \pi + \bar{\pi} = \pm 1\\
            \pi\bar{\pi} = p
        \end{cases} \Longleftrightarrow p = 3u^2 + 3u + 1,
    \]
    where \(u\in\mathbb{Z}\). Moreover, by \cite[Thm. III.10.1]{silvermanArithmeticEllipticCurves2009} and \( p^2 \equiv 1 \pmod{6} \), we can obtain that \(|\operatorname{Aut}(E, \mathcal{O})| = 6\). 

    (ii) In this case, the $j$-invariant \(j(\mathfrak{E})\) is equal to 1728 and the endomorphism ring of \( \mathfrak{E} \) is \({\rm End}(\mathfrak{E}) = \mathbb{Z}[\omega]\) where \(\omega = \sqrt{-1}\). Let \( \pi = u + v\omega\) with \( u,v\in\mathbb{Z} \). Then
    \[
        N(E) = p^2 + 2p - 3 \Longleftrightarrow t = \pi + \bar{\pi} = 2u = \pm 2.
    \]
    From the proposition of the Frobenius endomorphism, we have \( \pi\bar{\pi} = p \). Hence, after straightforward calculations, we conclude that
    \[
        N(E) = p^2 + 2p - 3 \Longleftrightarrow \begin{cases}
            \pi + \bar{\pi} = \pm 2\\
            \pi\bar{\pi} = p
        \end{cases} \Longleftrightarrow p = v^2 + 1,
    \]
    where \(v\in\mathbb{Z}\). Moreover, by \cite[Thm. III.10.1]{silvermanArithmeticEllipticCurves2009} and \( p^2 \equiv 1 \pmod{4} \), we can obtain that \(|\operatorname{Aut}(E, \mathcal{O})| = 4\).     The proof is completed.
\end{proof}

\begin{remark}\label{new_ellp_cuv_rem}
    \begin{itemize}
        \item[(i)] In the above proposition, both cases under consideration correspond to elliptic curves that have the largest number of rational points aside from maximal curves. For the elliptic curve \( \mathfrak{E}/ \mathbb{F}_{p^2} \) defined by the equation \( y^2 = x^3 + ax \), when \( N(E) = p^2 + 2p, p^2 + 2p - 1\) or \(p^2 + 2p - 2 \), the relations \( N(E) = p^2 + 2p + 1 - t^2 \) and \( t = 2u\) admit no integer solution for \( u \), and consequently, such a curve does not exist.
        \item[(ii)] It can be verified that \( p^2 + 2p - 3 \) is a square number if and only if \( p = 1 \). Hence, in each of the above two cases, there may exist a prime number \( s \) such that \( s \mid\mid N(E) \) (i.e. \( s\mid N(E)\) and \(\gcd(s,N(E)/s) = 1\)). Consequently, by employing a method similar to that in \cite{MA2023105686}, we can obtain new optimal locally repairable codes with more flexible localities. 
        \item[(iii)] By Mihăilescu's theorem\cite{metsankyla_catalans_2004}, there is no prime power \( q \) satisfying the condition \( q \equiv 1 \pmod{4} \) and \(q = v^2 + 1\) where \(v\in\mathbb{Z}\); however, it can be easily verified that \( 13^2 \) satisfies the condition \(13^2 = 3\cdot7^2 + 3\cdot7 + 1\). Therefore, Proposition \ref{new_ellp_cuv}(i) can be similarly applied to \( \mathfrak{E}/\mathbb{F}_{q^2} \) where \(q = 3u^2 + 3u + 1\) with \(u\in\mathbb{Z}\)
    \end{itemize}
\end{remark}

\section{Constructions of Locally Repairable Codes via Automorphism Groups}\label{construct} 

In this section, we divide the discussion into two parts. In the first part, we present an alternative method for determining the functions \( e_i \) in the construction of locally repairable codes, which is particularly useful for constructing locally repairable codes with two recovering sets in Section~\ref{construct_two}. In the second part, we introduce several new families of optimal classical locally repairable codes and provide some examples which seem unable to be obtained in \cite{Optimal_Locally_Repairable_Codes_Via_Elliptic_Curves} and \cite{MA2023105686}. . 

\subsection{Construction of Locally Repairable Codes via Automorphism Groups of Elliptic Function Fields}\label{lrc_framework}

Let $\mathfrak{E}/\mathbb{F}_q$ be an elliptic curve with function field $E$ and $\operatorname{Aut}(E/\mathbb{F}_q)$ be the automorphism group of $E$ over $\mathbb{F}_q$. Let $\mathcal{G}$ be a subgroup of $\operatorname{Aut}(E/\mathbb{F}_q)$ of order $r+1$ and $E^{\mathcal{G}}$ be the fixed subfield of $E$ with respect to $\mathcal{G}$. Then $E/E^{\mathcal{G}}$ is a Galois extension with Galois group $\operatorname{Gal}(E/E^{\mathcal{G}})=\mathcal{G}$.

Let $Q_1,Q_2,\cdots,Q_\ell$ be rational places of $E^{\mathcal{G}}$ which split completely in $E/E^{\mathcal{G}}$. Let $P_{i,1},P_{i,2},\cdots,P_{i,r+1}$ be the $r+1$ rational places of $E$ lying over $Q_i$ for $1\le i\le \ell$. Set $D = \sum_{i=1,j=1}^{\ell,r+1}P_{i,j}$. Let $e_1,e_2,\cdots,e_r$ be elements of $E$ that are linearly independent over $E^{\mathcal{G}}$ and satisfy
\[
    v_{P_{i,j}}(e_u)\ge0 \text{ for } 1\le i\le \ell,\,1\le j\le r+1 \text{ and }1\le u\le r.
\]
Let $f_1,f_2,\cdots,f_t$ be elements of $E^{\mathcal{G}}$ that are linearly independent over $\mathbb{F}_q$ and satisfy
\[
    v_{Q_i}(f_v)\ge0\text{ for }1\le i\le \ell\text{ and }1\le v\le t.
\]
Define
\begin{equation}\label{func}
    V=\Bigg\{ \sum_{j=1}^{t+1}a_{1,j}f_je_1+\sum_{i=2}^r\bigg(\sum_{j=1}^{t}a_{i,j}f_j\bigg)e_i:a_{i,j}\in\mathbb{F}_q \Bigg\}.
\end{equation}

The following lemma establishes result on local recoverability under certain conditions.

\begin{lemma}\label{local}\cite[Prop. 4.2]{Locally_Recoverable_Codes_from_Algebraic_Curves_and_Surfaces}
    Let $i$ be an integer with $1\le i\le \ell$, and suppose every $r\times r$ submatrix of the matrix 
    \[
        M=\begin{pmatrix}
        e_1(P_{i,1}) & e_2(P_{i,1}) & \cdots & e_r(P_{i,1}) \\
        e_1(P_{i,2}) & e_2(P_{i,2}) & \cdots & e_r(P_{i,2}) \\
        \vdots & \vdots & \ddots & \vdots \\
        e_1(P_{i,r+1}) & e_2(P_{i,r+1}) & \cdots & e_r(P_{i,r+1})
        \end{pmatrix}
    \]
    is invertible. Then the value of $f\in V$ at any place in the set $\{P_{i,1},P_{i,2},\cdots,P_{i,r+1}\}$ can be recovered from the values of $f$ at the other $r$ places.
\end{lemma}


From Lemma~\ref{local}, we see that the construction of the functions \( e_i \) is very important, as it guarantees the repair property (locality) of the linear code we construct. In the following, by considering certain subgroups of $\operatorname{Aut}(E/\mathbb{F}_q)$, we can determine the fixed subfield $E^{\mathcal{G}}$ and find a set of elements $e_1, e_2, \cdots, e_r \in E$ that are linearly independent over $E^{\mathcal{G}}$ and satisfy the property in Lemma~\ref{local}. 

\begin{proposition}\label{find_e_i}
    Let \( E/\mathbb{F}_q \) be an elliptic function field. Let
    \[
        H = \big\{P_1,P_2,\cdots,P_{|H|} = \mathcal{O}\big\}
    \]
    be a subgroup of \( \mathbb{P}^1_E \) with \( |H|<q \). Let \( T_H = \{\tau_P:P\in H\} \) be a subgroup of \( T_E \) and let \( A \) be a nontrivial subgroup of \( \operatorname{Aut}(E, \mathcal{O}) \) such that \( \mathcal{G} = T_H A \) is a subgroup of \( \operatorname{Aut}(E/\mathbb{F}_q) \) with \( |\mathcal{G}| = r + 1\). Then the following statements hold.

    \begin{itemize}
        \item[(i)] There exists an element \( z \in E \) satisfying that \( E^{\mathcal{G}} = \mathbb{F}_q(z) \) and \( (z)^E_\infty = |A|\sum_{i = 1}^{|H|} P_i \).
        \item[(ii)] There exist elements \( e_i \in E \) for \( 2\le i\le r \) satisfying that 
        \begin{equation}\label{e_i_condition}
            (e_i)^{E}_\infty = (\mu + 1)P_1 + (\mu + 1)P_2 + \cdots + (\mu + 1)P_\nu + \mu P_{\nu + 1} + \cdots + \mu P_{|H|}
        \end{equation}
        where \( i = \mu |H| + \nu \) with \( \mu,\nu\in\mathbb{Z}_{\ge0} \) and \( 0\le \nu < |H| \). Moreover, \( e_1 = 1, e_2, \cdots, e_r \) are linearly independent over \( E^{\mathcal{G}} \).
        \item[(iii)] Let \(\{P_{i,1}, P_{i,2}, \cdots, P_{i,r+1}\}\) be pairwise distinct rational places of \( E \) lying over a rational place \( Q_i \) of \( E^{\mathcal{G}} \) for each \( 1 \leq i \leq \ell \). Then all \( r \times r \) submatrices of the following matrix
        \[
            M=\begin{pmatrix}
            e_1(P_{i,1}) & e_2(P_{i,1}) & \cdots & e_r(P_{i,1}) \\
            e_1(P_{i,2}) & e_2(P_{i,2}) & \cdots & e_r(P_{i,2}) \\
            \vdots & \vdots & \ddots & \vdots \\
            e_1(P_{i,r+1}) & e_2(P_{i,r+1}) & \cdots & e_r(P_{i,r+1})
            \end{pmatrix}
        \]
        are invertible for all $1\le i\le \ell$. 
        \item[(iv)] For any place \(P_{i,j}\) with \(1\le j \le r+1\). Let 
        \[
            \{P_{i,1}^{(j)},P_{i,2}^{(j)},\cdots,P_{i,r+1 - |H|}^{(j)}\} = \{P_{i,1}, P_{i,2}, \cdots, P_{i,r+1}\}\backslash\{\tau_P(P_{i,j}) = P_{i,j}\oplus P:P\in H\}. 
        \]
        If \(P_{i,j}\notin E[|H|]\), then the following \( (r + 1 - |H|) \times (r + 1 - |H|) \) matrix 
        \[
            M^\prime=\begin{pmatrix}
            e_1(P_{i,1}^{(j)}) & e_2(P_{i,1}^{(j)}) & \cdots & e_{r + 1 -|H|}(P_{i,1}^{(j)}) \\
            e_1(P_{i,2}^{(j)}) & e_2(P_{i,2}^{(j)}) & \cdots & e_{r + 1 -|H|}(P_{i,2}^{(j)}) \\
            \vdots & \vdots & \ddots & \vdots \\
            e_1(P_{i,r+1 - |H|}^{(j)}) & e_2(P_{i,r+1 - |H|}^{(j)}) & \cdots & e_{r + 1 -|H|}(P_{i,r+1 - |H|}^{(j)})
            \end{pmatrix}
        \]
        is invertible. 
    \end{itemize}
\end{proposition}
\begin{proof}
    (i) In \cite[Prop. 4.2]{MA2023105686}, it has already been shown that \(E^{\mathcal{G}} = \mathbb{F}_q(z_0)\) with \( z_0\in E^{\mathcal{G}} \) and 
    \[
        (z_0)^{E}_0 = |A|\sum_{i = 1}^{|H|} P_i.
    \]
    We set \( z = z_0^{-1} \), then it follows that \( E^{\mathcal{G}} = \mathbb{F}_q(z_0) = \mathbb{F}_q(z) \) and \( (z)^E_\infty = |A|\sum_{i = 1}^{|H|} P_i \). 

    (ii) For each $2\le i\le r<q$ where \( i = \mu |H| + \nu \) with \( \mu,\nu\in\mathbb{Z}_{\ge0} \) and \( 0\le \nu < |H| \). Let
    \[
        D_i = (\mu + 1)P_1 + (\mu + 1)P_2 + \cdots + (\mu + 1)P_\nu + \mu P_{\nu + 1} + \cdots + \mu P_{|H|}. 
    \]
    By the Riemann-Roch Theorem\cite[Thm. 1.5.15]{Algebraic_Function_Fields_and_Codes}, we have
    \[
        \begin{cases}
            \left|\bigcup\limits_{j=1}^\nu\mathcal{L}(D_i-P_j)\right|\le \nu q^{i-1}<\left|\mathcal{L}(D_i)\right|=q^i,  &\mu = 0,\\
            \\
            \left|\bigcup\limits_{j=1}^{|H|}\mathcal{L}(D_i-P_j)\right|\le |H| q^{i-1}<\left|\mathcal{L}(D_i)\right|=q^i,  &\mu > 0.\\
        \end{cases}
    \]
    It implies that there exists an element \(e_i\) with \( (e_i)^{E}_\infty = D_i \) for each $2\le i\le r$. We now proceed to show that $e_1 = 1,e_2,\cdots,e_r$ are linearly independent over $E^{\mathcal{G}}$. Suppose, contrary to the claim, that a nontrivial relation $\sum_{i=1}^r \varphi_i(z) e_i = 0$ holds for some $\varphi_i(z) \in \mathbb{F}_q(z)$. Multiplying by a common denominator if necessary, we can assume without loss of generality that all $\varphi_i(z)$ are polynomials in $z$. Let $1\le w\le r$ with \( w = \mu_w |H| + \nu_w \), \( \mu_w,\nu_w\in\mathbb{Z}_{\ge0} \) and \( 0\le \nu_w < |H| \) be the largest integer such that 
    \[
        \deg(\varphi_w(z))=\mathop{\max}_{1\le i\le w}\{\,\deg(\varphi_i(z))\},
    \]
    and
    \[
        \deg(\varphi_w(z))>\mathop{\max}_{w< i\le r}\{\,\deg(\varphi_i(z))\,\}.
    \]
    \begin{itemize}
        \item If $w>1$, then for $i<w$ we have
    \begin{align*}
        v_{P_{\nu_w}}(\varphi_i(z)e_i)&=-|A|\deg(\varphi_i(z))+v_{P_{\nu_w}}(e_i)\\
        &\ge-|A|\deg(\varphi_i(z))-\mu_w\\
        &>-|A|\deg(\varphi_w(z))-\mu_w-1\\
        &=v_{P_{\nu_w}}(\varphi_w(z)e_w).
    \end{align*}
    For $i>w$, we have
    \begin{align*}
        v_{P_{\nu_w}}(\varphi_i(z)e_i)&=-|A|\deg(\varphi_i(z))+v_{P_{\nu_w}}(e_i)\\
        &\ge-|A|\deg(\varphi_i(z))-|A|\\
        &\ge-|A|\deg(\varphi_w(z))\\
        &>-|A|\deg(\varphi_w(z))+v_{P_{\nu_w}}(e_w)\\
        &=v_{P_{\nu_w}}(\varphi_w(z)e_w).
    \end{align*}
    The Strict Triangle Inequality \cite[Lem. 1.1.11]{Algebraic_Function_Fields_and_Codes} yields
    \[
        v_{P_{\nu_w}}\left(\sum_{i=1}^r\varphi_i(z)e_i\right)=v_{P_{\nu_w}}(\varphi_w(z)e_w)\neq v_{P_{\nu_w}}(0),
    \]
    which is a contradiction since $\sum_{i=1}^r\varphi_i(z)e_i=0$. 

        \item If $w=1$, then for $i>w$, we have
        \begin{align*}
            v_{P_{|H|}}(\varphi_i(z)e_i)&=-|A|\deg(\varphi_i(z))+v_{P_{|H|}}(e_i)\\
            &>-|A|\deg(\varphi_i(z))-|A|\\
            &\ge-|A|\deg(\varphi_w(z))+v_{P_{|H|}}(e_w)\\
            &=v_{P_{|H|}}(\varphi_w(z)e_w).
        \end{align*}
    It follows that
    \[
        v_{P_{|H|}}\left(\sum_{i=1}^r\varphi_i(z)e_i\right)=v_{P_{|H|}}(\varphi_w(z)e_w)\neq v_{P_{|H|}}(0),
    \]
    which is also a contradiction. 
    \end{itemize}
    Hence, the elements $e_1,\cdots,e_r$ are linearly independent over $E^{\mathcal{G}}$. 

    (iii) Let $P_{i,1},P_{i,2},\cdots,P_{i,r+1}$ be the $r+1$ rational places lying over $Q_i$ that corresponds to the zero of $z-\alpha_i$ for some $\alpha_i\in\mathbb{F}_q$. Without loss of generality, we can consider the first $r$ rows of the matrix $M$. Suppose that there exists ${\bf 0}\neq(c_1,\cdots,c_r)\in\mathbb{F}_q^r$ such that
    \[
        \begin{pmatrix}
        e_1(P_{i,1}) & e_2(P_{i,1}) & \cdots & e_r(P_{i,1}) \\
        e_1(P_{i,2}) & e_2(P_{i,2}) & \cdots & e_r(P_{i,2}) \\
        \vdots & \vdots & \ddots & \vdots \\
        e_1(P_{i,r}) & e_2(P_{i,r}) & \cdots & e_r(P_{i,r})
        \end{pmatrix}
        \begin{pmatrix}
        c_1 \\
        c_2 \\
        \vdots \\
        c_r
        \end{pmatrix}={\bf 0}.
    \]
    Set $f=c_1e_1+\cdots+c_re_r$. Then for $1\le j \le r$, we have 
    \[
        v_{P_{i,j}}(f)\ge1\quad {\rm i.e.}\quad \deg\, (f)^E_0\ge r.
    \]
    On the other hand, $f\in\mathcal{L}\big(|A|\sum_{i = 1}^{|H|-1} P_i + (|A|-1)P_{|H|}\big)$ i.e. $\deg(f)^E_{\infty}\le r$. Thus we derive that
    \begin{equation}\label{(f)}
        (f)^E=\sum_{j=1}^rP_{i,j}-\bigg(|A|\sum_{i = 1}^{|H|-1} P_i + (|A|-1)P_{|H|}\bigg).
    \end{equation}
    Note that $(z)_\infty^{\mathbb{F}_q(z)}=(z-\alpha_i)_\infty^{\mathbb{F}_q(z)}$, so 
    \begin{equation}\label{(z-a_i)}
        (z-\alpha_i)^E=\sum_{j=1}^{r+1}P_{i,j}-|A|\sum_{i = 1}^{|H|} P_i.
    \end{equation}
    By \eqref{(f)} and \eqref{(z-a_i)}, we conclude
    \[
        (f/(z-\alpha_i))^E=P_{|H|}-P_{i,r+1},
    \]
    which is a contradiction by Lemma \ref{P=Q}.

    (iv) Similar to the proof of (iii), suppose contrary to the statement, then we can derive that there exists a function \( g\in E \) with
    \[
        (g)^E = \sum_{t=1}^{r+1-|H|}P^{(j)}_{i,t}-\bigg((|A|-1)\sum_{i = 1}^{|H|} P_i\bigg), 
    \]
    and therefore
    \[
         (g/(z-\alpha_i))^E=\sum_{t=1}^{|H|} (P_{i,j}\oplus P_t) - \sum_{i = 1}^{|H|} P_i. 
    \]
    It follows that
    \[
        \oplus_{t=1}^{|H|} (P_{i,j}\oplus P_t) = ([|H|]P_{i,j})\oplus (\oplus_{t = 1}^{|H|} P_t) = \oplus_{i = 1}^{|H|} P_i\in \mathbb{P}_E^1,
    \]
    where \([|H|]P_{i,j}\) means that \(\underbrace{P_{i,j}\oplus\cdots\oplus P_{i,j}}_{|H|\text{ terms}}\). So we have
    \[
        [|H|]P_{i,j} = \mathcal{O}\in \mathbb{P}_E^1,
    \]
    which contradicts to \(P_{i,j}\notin E[|H|]\). 

    The proof is completed. 
\end{proof}

Let \( N(E) \) be the number of rational places of \( E/\mathbb{F}_q \) and let \(T_H\) and \(A\) be given in Proposition \ref{find_e_i}. Then we obtain that \( E^{\mathcal{G}} = \mathbb{F}_q(z) \) with \( (z)^E_\infty = |A|\sum_{i = 1}^{|H|} P_i \) and find a set of elements $e_1, e_2, \cdots, e_r \in E$. There are at least \( \ell := \left\lceil \frac{N(E)-2|T|}{r+1} - 1 \right\rceil = \left\lceil \frac{N(E)-2|T|}{r+1} \right\rceil - 1 \) rational places of \( E^{\mathcal{G}} \) which split completely in \( E/E^{\mathcal{G}} \) from Lemma \ref{split_place}. Set \(f_i = z^{i-1}\) for any \(1\le i\le t\). Then the function space \(V\) defined by \eqref{func} is a subspace of \( \mathcal{L}\big(t|A|\sum_{i = 1}^{|H|} P_i\big)\). Let \( 1 \le t < m \le \ell \) and \( D = \sum_{i=1,j=1}^{m,r+1}P_{i,j} \). Then the algebraic geometry code
\[
    C_{\mathcal{L}}(D,V)=\big\{(f(P))_{P\in\operatorname{supp}(D)}:f\in V\big\}
\]
is an \([n = m(r+1), k = rt+1, d \ge n-t(r+1)]_q\) locally repairable code with locality \( r \). Indeed, the locality of $C_{\mathcal{L}}(D,V)$ follows directly from Lemma \ref{local} together with Proposition \ref{find_e_i}(iii). 

\begin{remark}
    In Proposition 20 in \cite{Optimal_Locally_Repairable_Codes_Via_Elliptic_Curves} and Proposition 4.2 in \cite{MA2023105686}, the authors use a set of completely splitting places to determine the functions \( e_i \). Subsequently, they ingeniously employ the modified algebraic geometry code technique to compensate for the absence of such completely splitting places. Here, we adopt a different approach. We fully exploit the properties of subgroup \( \mathcal{G} \) to determine the functions \( e_i \). Moreover, the functions \( e_i \) we obtain also satisfy the property in Lemma~\ref{local} as those in \cite{Optimal_Locally_Repairable_Codes_Via_Elliptic_Curves} and \cite{MA2023105686}. 
\end{remark}

\subsection{New Families of Optimal Classical Locally Repairable Codes via Elliptic Function Fields}

In this subsection, we will employ the aforementioned construction along with Proposition \ref{new_ellp_cuv} to present several new families of optimal locally repairable codes via elliptic function fields. It seems that the parameters of the optimal locally repairable codes we construct are not attainable by previous constructions in the literature. At first, we provide a proof of Theorem~\ref{opt_LRC_one_hA}. 
\\

{
\raggedright
{\it Proof of Theorem~\ref{opt_LRC_one_hA}.} 
}By Proposition \ref{new_ellp_cuv}, there exist elliptic curves satisfying the above conditions about \(|A|\), \( p \) and \(N\). Hence, there exists a subgroup of \( \operatorname{Aut}(E/\mathbb{F}_q) \) with order \(h|A|\) from Lemma \ref{sub} and \(h\mid\mid N\). Following the discussion in Section~\ref{lrc_framework}, we can obtain an \([n = m(r+1), k = rt+1, d \ge (m-t)(r+1)]_{p^2}\) locally repairable code with locality \( r \) via elliptic function field. 

    On the other hand, applying the Singleton-type upper bound \eqref{min_d_bound_one} yields 
    \begin{align*}
        d &\leq n - k - \left\lceil \frac{k}{r} \right\rceil + 2 \\
        & = n - rt - 1 - \left\lceil \frac{tr + 1}{r} \right\rceil + 2 \\
        &= n - t(r + 1) = (m-t)(r+1).
    \end{align*}
    Therefore, we conclude that the code we construct is indeed an optimal locally repairable code.\qed
\\

Let \(T\) be a subgroup of the translation group \(T_E\) and let \(A\) be the cyclic subgroup of \(\operatorname{Aut}(E, \mathcal{O})\) generated by \(\sigma\), where \(\sigma\) is defined by \(\sigma(x) = x\) and
\[
    \sigma(y) =
    \begin{cases} 
        -y, & \text{char}(\mathbb{F}_q) \neq 2, \\
        y + 1, & \text{char}(\mathbb{F}_q) = 2 \text{ and } j(E) = 0, \\
        y + x, & \text{char}(\mathbb{F}_q) = 2 \text{ and } j(E) \neq 0.
    \end{cases}
\]
By Proposition 4.5 of \cite{MA2023105686}, the product \(TA\) forms a subgroup of \(\operatorname{Aut}(E/\mathbb{F}_q)\) of order \(2|T|\). From this, we can provide a proof of Theorem~\ref{opt_LRC_one_2h}.
\\

{
\raggedright
{\it Proof of Theorem~\ref{opt_LRC_one_2h}.} 
}By Proposition \ref{new_ellp_cuv}, there exist elliptic curves satisfying the above conditions about \(|A|\), \( p \) and \(N\). Let \(T\) be a subgroup of the translation group \(T_E\) with order \(h\). Then from the above discussion, it follows that there exists an subgroup of \( \operatorname{Aut}(E/\mathbb{F}_q) \) with order \(2h\). The rest of the proof follows analogously to that of Theorem \ref{opt_LRC_one_hA}. \qed

\begin{remark}\label{opt_LRC_one_re}
    \begin{itemize}
        \item[(i)] From Remark \ref{new_ellp_cuv_rem}(iii), for a prime power \(q\), there also exists an optimal \( q^2 \)-ary \([m(r+1), rt+1, (m-t)(r+1)]_{q^2}\) locally repairable code with locality \( r = h|A| - 1 \) or \( r = 2h - 1 \) provided that \(|A|\), \( q \) and \(N\) satisfy \(N = q^2 + 2q\) and \(|A| = 2,3,6\) for \(q = 3u^2 + 3u + 1\) with \(u\in\mathbb{Z}\). 
        \item[(ii)] It can be seen that the length of our \(q\)-ary optimal locally repairable codes constructed in Theorems~\ref{opt_LRC_one_hA} and \ref{opt_LRC_one_2h} can approach \(q + 2\sqrt{q}\). 
        \item[(iii)] Reviewing the results of optimal locally repairable codes in \cite{Optimal_Locally_Repairable_Codes_Via_Elliptic_Curves} and \cite{MA2023105686}, since their construction is based on maximal elliptic function fields, the resulting codes can only be defined over finite fields with characteristic \(p=2\), \(3\), \( p \equiv 3 \pmod{4} \) and \(p \equiv 2 \pmod{3}\). Theorems \ref{opt_LRC_one_hA} and \ref{opt_LRC_one_2h} take into account the case where the characteristic \(p\) of the finite field is both \( p \equiv 1 \pmod{3} \) and \( p \equiv 1 \pmod{4} \), a case for which no constructions of optimal locally repairable codes were available in \cite{Optimal_Locally_Repairable_Codes_Via_Elliptic_Curves} and \cite{MA2023105686}. In addition, even if over the same finite field, the locality of our optimal locally repairable codes seems to be able to differ from those in \cite{Optimal_Locally_Repairable_Codes_Via_Elliptic_Curves} and \cite{MA2023105686}. Therefore, our results effectively complement previous work. 
    \end{itemize}
\end{remark}

In the following, we present several examples of our optimal locally repairable codes which seem unable to be obtained in \cite{Optimal_Locally_Repairable_Codes_Via_Elliptic_Curves} and \cite{MA2023105686}. 

\begin{example}\label{opt_lrc_one_ex}
    \begin{itemize}
        \item[(i)] Let \(p=7 = 3\cdot 1^2 + 3\cdot 1 + 1\) and \(\mathbb{F}_{p^2} = \mathbb{F}_p(\omega)\). Then we can find the elliptic curve \(\mathfrak{E}/\mathbb{F}_{p^2}:y^2 = x^3 + 2\) with \(N(E) = p^2 + 2p = 63\). Take \(H = E[7]\) with \(|H| = 7\mid 63\) and \(A = \{\sigma_u:u^2 = 1\}\) where \(\sigma_u(x) = u^2x\) and \(\sigma_u(y) = u^3y\). Let \(\mathcal{G} = T_HA\). Then we obtain that \(E^{\mathcal{G}} = \mathbb{F}_{p^2}(z)\) where \(z = (x^7 + 3x^4 + 6x)/(x^6 + 5x^3 + 1)\). By Theorem~\ref{opt_LRC_one_2h}, we can construct an optimal \(7^2\)-ary 
        \[
            [n=14m,\,k=13t+1,\,14(m-t)]_{7^2}
        \]
        locally repairable code with locality \(r = 13\) for any \(1\le t\le m\le4\). In Theorems~1.1 and 1.2 of \cite{MA2023105686}, optimal locally repairable codes with locality \(r = 2^\ell-1\) and \(2\le\ell\le8\) over \(\mathbb{F}_{7^2}\) can be constructed. Hence, it seems that our optimal locally repairable code with locality \(r = 13\) does not seem to be obtainable from \cite{MA2023105686}.
        \item[(ii)] Let \(p=37 = 6^2 + 1\) and \(\mathbb{F}_{p^2} = \mathbb{F}_p(\omega)\). Then we can find the elliptic curve \(\mathfrak{E}/\mathbb{F}_{p^2}:y^2 = x^3 + x\) with \(N(E) = p^2 + 2p - 3 = 1440\). Take \(H = E[5]\) with \(|H| = 5\mid\mid 1440\) and \(A = \operatorname{Aut}(E, \mathcal{O}) = \{\sigma_u:u^4 = 1\}\) where \(\sigma_u(x) = u^2x\) and \(\sigma_u(y) = u^3y\). Let \(\mathcal{G} = T_HA\). Then we obtain that \(E^{\mathcal{G}} = \mathbb{F}_{p^2}(z)\) where \(z = (11x^{10} + 7x^8 + 24x^6 + 13x^4 + 30x^2 + 26)/(x^8 + 3x^6 + 8x^4 + 4x^2 + 10)\). By Theorem~\ref{opt_LRC_one_hA}, we can construct an optimal \(37^2\)-ary 
        \[
            [n=20m,\,k=19t+1,\,20(m-t)]_{37^2}
        \]
        locally repairable code with locality \(r = 19\) for any \(1\le t\le m\le71\). We note that optimal locally repairable codes over over finite field of characteristic \(37\) do not seem to be obtained in \cite{Optimal_Locally_Repairable_Codes_Via_Elliptic_Curves} and \cite{MA2023105686}.
        \item[(iii)] Let \(q=13^2 = 3\cdot7^2 + 3\cdot 7+1\) and \(\mathbb{F}_{q^2} = \mathbb{F}_q(\omega)\). Then we can find the elliptic curve \(\mathfrak{E}/\mathbb{F}_{q^2}:y^2 = x^3 + 2\) with \(N(E) = q^2 + 2q = 28899\). Take \(H = E[9]\) with \(|H| = 9\mid 28899\) and \(A = \{\sigma_u:u^2 = 1\}\) where \(\sigma_u(x) = u^2x\) and \(\sigma_u(y) = u^3y\). Let \(\mathcal{G} = T_HA\). Then we obtain that \(E^{\mathcal{G}} = \mathbb{F}_{p^2}(z)\) where \(z = (3x^9 + 7x^8 + 9x^6 + 8x^5 + 4x^3 + 6x^2 + 2)/(x^8 + 3x^5 + 12x^2)\). By Theorem~\ref{opt_LRC_one_2h}, we can construct an optimal \(13^4\)-ary 
        \[
            [n=18m,\,k=17t+1,\,18(m-t)]_{13^4}
        \]
        locally repairable code with locality \(r = 17\) for any \(1\le t\le m\le1605\). We note that optimal locally repairable codes over over finite field of characteristic \(13\) do not seem to be obtained in \cite{Optimal_Locally_Repairable_Codes_Via_Elliptic_Curves} and \cite{MA2023105686}.
    \end{itemize}

    From the above examples, we can see that our results enrich the existing results on optimal locally repairable codes, i.e. our results not only produce \(q\)-ary optimal locally repairable codes with length \(O(q+2\sqrt{q})\) over more finite fields, but also offer more flexible locality. 
\end{example}

\section{Locally Repairable Codes with Two Recovering Sets via Elliptic Function Fields}\label{construct_two} 

In this section, we focus on the construction of locally repairable codes with two recovering sets via elliptic function fields. We first present a general framework for such constructions, and then, utilizing maximal elliptic function fields and Proposition \ref{new_ellp_cuv}, we construct several distinct families of locally repairable codes equipped with two recovering sets. Furthermore, we present a special construction of locally repairable codes with two recovering sets on a specific elliptic function field defined over a finite field of characteristic \(2\). 

\subsection{General Framework for Constructing Locally Repairable Codes with Two Recovering Sets}

We first present a general framework for constructing locally repairable codes with two recovering sets via elliptic function fields. 

\begin{proposition}\label{cons_lrc_two}
    Let \( E/\mathbb{F}_q \) be an elliptic function field. Let \(H = \{P_1,P_2,\cdots,P_{|H|} = \mathcal{O}\}\) be a subgroup of \( \mathbb{P}^1_E \) with \( |H|<q \). Let \( T_H = \{\tau_P:P\in H\} \) be a subgroup of \( T_E \) and let \( A_1,A_2 \) be two nontrivial subgroups of \( \operatorname{Aut}(E, \mathcal{O}) \) such that \(|A_1\cap A_2| = 1\) and \( \mathcal{G} = T_H A_1 A_2 \) is a subgroup of \( \operatorname{Aut}(E/\mathbb{F}_q) \). Suppose that \( Q_1, \cdots, Q_m \) are rational places of \( E^{\mathcal{G}} \) that split completely in \( E \) and the rational place $P_{i,j}$ of $E$ lying over $Q_i$ satisfies that \(P_{i,j}\notin E[|H|]\) for any $1\le i\le m$ and \(1\le j\le |H||A_1||A_2|\). Put \( r_1 = |H||A_1| - 1 \), \( r_2 = |H|(|A_2| - 1)\) and \( n = m|H||A_1||A_2| \). Then for any integer \( d_0 \) with \( 1 \leq d_0 < n \), there exists a \( q \)-ary \([n, k, d; (r_1, r_2)]_q\) locally repairable code with \(d \ge n - L\ge d_0\) and  
    \[
        k \geq \bigg\lfloor\frac{n-d_0}{r_1 + 1}\bigg\rfloor r_1 + \bigg\lfloor\frac{n-d_0}{r_2 + |H|}\bigg\rfloor r_2  + 2 - L,
    \]
    where
    \[
        L =  \max\left\{\left\lfloor\frac{n-d_0}{r_1 + 1}\right\rfloor (r_1+1),\left\lfloor\frac{n-d_0}{r_2 + |H|}\right\rfloor (r_2 + |H|)\right\}.  
    \]
\end{proposition}
\begin{proof}
    By a standard group-theoretic result, we have
    \[
        |\mathcal{G}| = \frac{|T_H|\cdot|A_1| \cdot |A_2|}{|T_H\cap (A_1A_2)|\cdot|A_1 \cap A_2|} = |H|\cdot|A_1| \cdot |A_2|.
    \]
    It follows that the extension \( E/E^{\mathcal{G}} \) has degree \(|H||A_1| |A_2|\). Consider the field extensions tower in Fig.\ref{extension}. 
    \begin{figure}[htbp]
        \centering
        \begin{tikzcd}
          & E \ar[dl, -] \ar[dr, -] & \\
         E^{T_H A_1} \ar[dr, -] & & E^{T_H A_2} \ar[dl, -] \\
          & E^{\mathcal{G}} &
        \end{tikzcd}
        \caption{The field extensions tower}
        \label{extension}
    \end{figure}
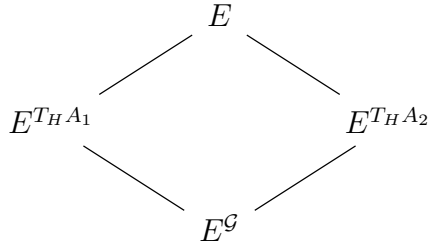
    
    By Proposition \ref{find_e_i}, we derive that \(E^{T_H A_i} = \mathbb{F}_q(z_i)\) with \((z_i)^E_\infty = |A_i|\sum_{u = 1}^{|H|} P_u\) for \(i = 1,2\). Moreover, we can also derive that, for \(i = 1,2\), there exist elements \( e_\ell^{(i)} \in E \) for \( 2\le \ell\le r_i \) satisfying \eqref{e_i_condition} and \( e^{(i)}_1 = 1, e^{(i)}_2, \cdots, e^{(i)}_{r_i} \) are linearly independent over \( E^{T_HA_i} \). Define 
    \[
        V_i = \Bigg\{ \sum_{j=0}^{t_i}a_{1,j}z_i^j+\sum_{\ell=2}^{r_i}\bigg(\sum_{j=0}^{t_i-1}a_{\ell,j}z_i^j\bigg)e^{(i)}_\ell:a_{\ell,j}\in\mathbb{F}_q \Bigg\}
    \]
    with
    \[
        t_1 = \left\lfloor\frac{n-d_0}{r_1 + 1}\right\rfloor \text{ and }t_2 = \left\lfloor\frac{n-d_0}{r_2 + |H|}\right\rfloor. 
    \]
    Then it follows that 
    \[
        V_i \subseteq \mathcal{L}\left( L_0 \sum_{u = 1}^{|H|} P_u\right),  
    \]
    where \(L_0 = \max\left\{\left\lfloor\frac{n-d_0}{r_1 + 1}\right\rfloor |A_1|,\left\lfloor\frac{n-d_0}{r_2 + |H|}\right\rfloor |A_2|\right\}\). Let \(V = V_1\cap V_2\). Then we have
    \begin{align*}
        \dim_{\mathbb{F}_q} V &= \dim_{\mathbb{F}_q}V_1 + \dim_{\mathbb{F}_q}V_2 - \dim_{\mathbb{F}_q}(V_1 + V_2) \\ 
        &\ge \dim_{\mathbb{F}_q}V_1 + \dim_{\mathbb{F}_q}V_2  - \dim_{\mathbb{F}_q}\mathcal{L}\left( L_0 \sum_{u = 1}^{|H|} P_u\right) \\ 
        &=  \bigg\lfloor\frac{n-d_0}{r_1 + 1}\bigg\rfloor r_1 + \bigg\lfloor\frac{n-d_0}{r_2 + |H|}\bigg\rfloor r_2 + 2  - L, 
    \end{align*}
    where
    \[
        L =  \max\left\{\left\lfloor\frac{n-d_0}{r_1 + 1}\right\rfloor (r_1+1),\left\lfloor\frac{n-d_0}{r_2 + |H|}\right\rfloor (r_2 + |H|)\right\}.  
    \]
    Set \( D = \sum_{i=1,j=1}^{m,|H||A_1||A_2|}P_{i,j} \). Then we can easily obtain an \([n, k, d]_q\) algebraic geometry code
    \[
       C_{\mathcal{L}}(D,V)=\big\{(f(P))_{P\in\operatorname{supp}(D)}:f\in V\big\}. 
    \]
    
    In the following, we will show that this code is indeed a \( q \)-ary \([n, k, d; (r_1, r_2)]_q\) locally repairable code. Let $P$ be a rational place in $\operatorname{supp}(D)$ lying over \(Q\in\mathbb{P}_{E^\mathcal{G}}^1\). For each $i = 1, 2$, there exists a rational place $R_i = P\cap E^{T_H A_i}$ of $E^{T_H A_i}$ such that $P$ lies over $R_i$. Denote by $\mathcal{P}_i$ the set of places of $E$ lying over $R_i$. Then $P$ belongs to both $\mathcal{P}_1$ and $\mathcal{P}_2$. Since \(T_HA_1\cap T_HA_2 = T_H\), we conclude that 
    \[
        \mathcal{P}_1\cap \mathcal{P}_2 = \left\{\tau_{P^\prime}(P) = P\oplus P^\prime:P^\prime\in H\right\}. 
    \]
    Now, by Lemma \ref{local} and Proposition \ref{find_e_i}(iii), it follows that for any \(f \in V\), the value \(f(P)\) can be repaired from the set \(\{f(P') : P' \in \mathcal{P}_1 \setminus \{P\}\}\). This means that the code we constructed has locality \( r_1 \). Moreover, for every \(P^\prime\in \mathcal{P}_2\setminus(\mathcal{P}_1\cap \mathcal{P}_2)\), we have
    \begin{align*}
        f(P^\prime) &= \sum_{j=0}^{t_2}a_{1,j}z_2^j(P^\prime)+\sum_{\ell=2}^{r_2}\bigg(\sum_{j=0}^{t_2-1}a_{\ell,j}z_2^j(P^\prime)\bigg)e^{(2)}_\ell(P^\prime) \\
        & = \sum_{j=0}^{t_2}a_{1,j}z_2^j(R_2)+\sum_{\ell=2}^{r_2}\bigg(\sum_{j=0}^{t_2-1}a_{\ell,j}z_2^j(R_2)\bigg)e^{(2)}_\ell(P^\prime)\\
        & = \sum_{\ell=1}^{r_2} c_\ell e^{(2)}_\ell(P^\prime)
    \end{align*}
    where
    \[
        c_\ell = \begin{cases}
            \sum_{j=0}^{t_2}a_{1,j}z_2^j(R_2), &\ell = 1\\
            \sum_{j=0}^{t_2-1}a_{\ell,j}z_2^j(R_2), & \ell>1 
        \end{cases}. 
    \]
    Since the values of \(e^{(2)}_\ell(P')\) are known for each \(\ell\) and \(P'\), we can determine the coefficients \(c_\ell\) from the values of \(f\) at the other \(r_2\) places in \(\mathcal{P}_2 \setminus (\mathcal{P}_1 \cap \mathcal{P}_2)\), using the assumption that \(P \notin E[|H|]\) together with Proposition \ref{find_e_i}(iv). Hence, the code \( C_{\mathcal{L}}(D,V) \) is indeed a locally repairable code with localities \((r_1, r_2)\) and availability \( 2 \). The proof is completed. 
\end{proof}

\begin{remark}
    It is worth emphasizing that if we adopt the functions \( e_i \) and \( z \) from the framework in \cite{MA2023105686}, then the intersection of the two function spaces \( V_1 \) and \(V_2\) in Proposition~\ref{cons_lrc_two} will be very small, which limits the parameter flexibility of the resulting locally repairable codes with two recovering sets. This is also the key motivation for a novel construction of the functions \(e_i\) presented in this paper. 
\end{remark}

Using the above proposition together with specific elliptic function fields, we can construct several distinct families of locally repairable codes equipped with two recovering sets as illustrated in Theorems~\ref{lrc_two_max} and \ref{lrc_two_nonmax}. In the following we restate Theorem~\ref{lrc_two_max}, which presents the result constructed on maximal elliptic function fields, and provide its proof. 

\begin{theorem}\label{lrc_two_max_restate}
    Let \( q = p^a \) for any prime \( p \) and any even integer \( a > 0 \). Let \(\sqrt{q}+1 = \prod_\ell \ell^{h_\ell}\)  be the prime factorization of \( \sqrt{q} + 1 \). For any integer \( d_0 \) with \( 1 \leq d_0 < n \) and any positive divisor \( h \) of \( \sqrt{q} + 1 \) with
    \begin{equation}\label{lrc_two_repair_condition}
        h^2|A_1||A_2|>\prod_\ell \ell^{2\min\{2\nu_\ell(h),h_\ell\}}-h^2,
    \end{equation}
    there exists a \( q \)-ary \([n = mh^2|A_1||A_2|, k, d;(r_1,r_2)]_q\) locally repairable code where \( r_1 = h^2|A_1| - 1 \), \(r_2 = h^2(|A_2|-1)\), \(d \ge n - L \ge d_0\), 
    \[
        k \geq \bigg\lfloor\frac{n-d_0}{r_1 + 1}\bigg\rfloor r_1 + \bigg\lfloor\frac{n-d_0}{r_2 + h^2}\bigg\rfloor r_2  + 2 - L,
    \]
    and
    \[
        L = \max\left\{\left\lfloor\frac{n-d_0}{r_1 + 1}\right\rfloor (r_1+1),\left\lfloor\frac{n-d_0}{r_2 + h^2}\right\rfloor (r_2 + h^2)\right\}, 
    \]
    for any integer \( m \) satisfying \(1 \leq m \leq \left\lceil \frac{q+2\sqrt{q}+1-2h^2}{h^2|A_1||A_2|} \right\rceil - 1 \), provided that \(|A_1|\), \(|A_2|\) and \( p \) satisfy one of the following cases:
     
    \begin{itemize}
        \item[(i)] \(|A_1| = 2,4,8\) and \(|A_2| = 3\) for \( p = 2 \);  
        \item[(ii)] \(|A_1| = 2,4\) and \(|A_2| = 3\) for \( p = 3 \);  
        \item[(iii)] \(|A_1| = 2\) and \(|A_2| = 3\) for \( p \equiv 2 \pmod{3} \) and \( p \neq 2 \).
    \end{itemize}
\end{theorem}
\begin{proof}
    For a maximal elliptic function field \(E/\mathbb{F}_q\), the translation group \(T_E\) of \(E\) has the following group structure:
    \[
    T_E \cong \mathbb{Z}/(\sqrt{q} + 1)\mathbb{Z} \times \mathbb{Z}/(\sqrt{q} + 1)\mathbb{Z}.
    \]
    Let \(H = E[h]\) with \(|H| = h^2\). Then, by \cite{Optimal_Locally_Repairable_Codes_Via_Elliptic_Curves} and \cite[Thm. 1.2]{MA2023105686}, we can derive the existence of subgroups \(T_H\), \(A_1\) and \(A_2\) that satisfy one of the above cases together with the conditions in Proposition \ref{cons_lrc_two}. Now we only need to prove that for any rational place \(Q\) of \( E^{T_HA_1A_2} \) that split completely in \( E \), the rational place $P$ of $E$ lying over $Q$ satisfies that \(P\notin E[h^2]\); then, together with Proposition \ref{cons_lrc_two}, we can obtain our conclusion.

    Suppose, on the contrary, that \(P\in E[h^2]\). Then we have \([h^2]P = \mathcal{O}\in \mathbb{P}_E^1\). It follows that \(h^2P\sim h^2\mathcal{O}\). For any \(\tau\sigma\in T_HA_1A_2\), we can conclude that
    \begin{align*}
        h^2P\sim h^2\mathcal{O} &\Longleftrightarrow h^2\tau\sigma(P)\sim h^2\tau\sigma(\mathcal{O})\\
        &\Longleftrightarrow h^2\tau\sigma(P)\sim h^2\tau(\mathcal{O}).
    \end{align*}
    Note that \(\tau(\mathcal{O})\in H\). Hence we have \([h^2]\tau(\mathcal{O}) = \mathcal{O}\) which follows that \(h^2\tau(\mathcal{O})\sim h^2\mathcal{O}\). Therefore \(h^2\tau\sigma(P)\sim h^2\mathcal{O}\) i.e. \([h^2]\tau\sigma(P) = \mathcal{O}\). This implies that we can obtain
    \[
         h^2|A_1||A_2|\le |E[h^2]|-|E[h]|,
    \]
    i.e.
    \[
        h^2|A_1||A_2|\le\prod_\ell \ell^{2\min\{2\nu_\ell(h),h_\ell\}}-h^2
    \]
    which contradicts to the assumption. Hence, we have \(P\notin E[h^2]\). The proof is completed. 
\end{proof}

In the following we restate Theorem~\ref{lrc_two_nonmax}, which presents the result obtained from elliptic function fields corresponding to the elliptic curves in Proposition \ref{new_ellp_cuv}, and provide its proof. 

\begin{theorem}
    Let \( q \) be a prime power. For any integer \( d_0 \) with \( 1 \leq d_0 < n \) and any positive divisor \( h \) with \(h\mid\mid q^2+2q\), there exists a \( q^2 \)-ary \([n = mh|A_1||A_2|, k, d;(r_1,r_2)]_{q^2}\) locally repairable code where \( r_1 = h|A_1| - 1 \), \(r_2 = h(|A_2|-1)\), \(d \ge n - L \ge d_0\), 
    \[
        k \geq \bigg\lfloor\frac{n-d_0}{r_1 + 1}\bigg\rfloor r_1 + \bigg\lfloor\frac{n-d_0}{r_2 + h}\bigg\rfloor r_2  + 2 - L,
    \]
    and
    \[
        L = \max\left\{\left\lfloor\frac{n-d_0}{r_1 + 1}\right\rfloor (r_1+1),\left\lfloor\frac{n-d_0}{r_2 + h}\right\rfloor (r_2 + h)\right\} , 
    \]
     for any integer \( m \) satisfying \(1 \leq m \leq \left\lceil \frac{q^2+2q-2h}{h|A_1||A_2|} \right\rceil - 1 \), provided that \(|A_1|\), \(|A_2|\) and \( q \) satisfy that \(|A_1| = 2\) and \(|A_2| = 3\) for \(q = 3u^2 + 3u + 1\) with \(u\in\mathbb{Z}\). 
\end{theorem}
\begin{proof}
    Apart from using the elliptic curve in Proposition \ref{new_ellp_cuv}(i) and the corresponding elliptic function field, the proof is analogous to that of Theorem \ref{lrc_two_max_restate}. 
    
    Let \(H = E[h]\). Then, by Proposition \ref{new_ellp_cuv}(i), we can derive \(|H| = h\) and the existence of subgroups \(T_H\), \(A_1\) and \(A_2\) that satisfy the above assumptions together with the conditions in Proposition \ref{cons_lrc_two}. Therefore, we only need to prove that for any rational place \(Q\) of \( E^{T_HA_1A_2} \) that split completely in \( E \), the rational place $P$ of $E$ lying over $Q$ satisfies that \(P\notin E[|H|]\); then, together with Proposition \ref{cons_lrc_two}, we can obtain our conclusion. 

    Suppose, on the contrary, that \(P\in E[|H|]\). Then we can also obtain 
    \[
        h|A_1||A_2|\le |E[|H|]|-|E[h]|.
    \]
    However, it can be readily verified that \(|E[|H|]|=|E[h]|=h\). So we have \(h|A_1||A_2|\le 0 \) which is a contradiction. Hence, we have \(P\notin E[|H|]\). The proof is completed. 
\end{proof}

\begin{remark}\label{lrc_two_re}
    \begin{itemize}
        \item[(i)] From Proposition~\ref{cons_lrc_two}, we can see that the condition for subgroups \(A_1\) and \( A_2 \) is \( |A_1\cap A_2| = 1 \). Therefore, the roles of subgroups \( A_1 \) and \( A_2 \) in Theorems~\ref{lrc_two_max} and \ref{lrc_two_nonmax} can be interchanged.
        \item[(ii)] In fact, the condition \eqref{lrc_two_repair_condition} can be omitted when constructing locally repairable codes equipped with two recovering sets. However, in this case, we are unable to consider places in \( E[|H|] \) to construct algebraic geometry codes, and as a result, the lengths of the locally repairable codes that we construct become shorter.
    \end{itemize}
\end{remark}

Now we present families of examples of our locally repairable codes with two recovering sets. Moreover, we obtain a good estimate of the Singleton-defect \( \Delta(\mathcal{C}) \) of our locally repairable codes. 

\begin{example}\label{lrc_two_ex_1}
    Let \( q \) be a prime power satisfying that \(q = 3u^2 + 3u + 1\) with \(u\in\mathbb{Z}\). 
    \begin{itemize}
        \item[(i)] Let \(h = q\) and \(m = \left\lceil \frac{q}{6} \right\rceil - 1\). Therefore \(n = 6q\left(\left\lceil \frac{q}{6} \right\rceil - 1\right)\). Let \(|A_1| = 2\), \(|A_2| = 3\) and \(d_0 = n - 5q\). Then by Theorem~\ref{lrc_two_nonmax}, there exists a \( q^2 \)-ary \([n, k, \ge n - 4q;(2q-1,2q)]_{q^2}\) locally repairable code \(\mathcal{C}\), where \(k\) satisfies
        \[
            k \geq 2(2q-1) + 2q  + 2 - 4q = 2q. 
        \]
        The Singleton-defect of the code \(\mathcal{C}\) is
        \[
            \Delta(\mathcal{C}) \le \frac{1}{q^2-6q}(n - 2q - n + 4q + 1) = \frac{2q + 1}{q^2-6q}. 
        \]
        Consequently, the Singleton-defect tends to zero as \(q\) goes to infinity.
        \item[(ii)] Let \(h = \ell \mid\mid q + 2\) with \(4\ell<q\). Let \(m = \left\lceil \frac{q^2 +2q - 2\ell}{6\ell} \right\rceil - 1\). Therefore \(n = 6\ell\left(\left\lceil \frac{q^2 +2q - 2\ell}{6\ell} \right\rceil - 1\right)\). Let \(|A_1| = 3\), \(|A_2| = 2\) and \(d_0 = n - 3\ell\). Then by Theorem~\ref{lrc_two_nonmax}, there exists a \( q^2 \)-ary \([n, k, \ge n - 3\ell;(3\ell-1,\ell)]_{q^2}\) locally repairable code \(\mathcal{C}\), where \(k\) satisfies
        \[
            k \geq (3\ell-1) + \ell  + 2 - 3\ell = \ell + 1. 
        \]
        The Singleton-defect of the code \(\mathcal{C}\) is
        \[
            \Delta(\mathcal{C}) \le \frac{1}{q^2 + 2q - 8\ell}(n - (\ell+1) - n + 3\ell + 1) = \frac{2\ell}{q^2 + 2q - 8\ell}. 
        \]
        Similarly, the Singleton-defect tends to zero as \(q\) goes to infinity. Note that
        \[
            n\ge 6\ell\left( \frac{q^2 +2q - 2\ell}{6\ell}  - 1\right) = q^2 + 2q - 8\ell > q^2. 
        \]
        Therefore, as \( q \) increases, we can construct \( q^2 \)-ary locally repairable codes with length \(O(q^2 + 2q)\) and small Singleton-defect.
    \end{itemize}
\end{example}

In the following we present some numerical examples of our locally repairable codes with two recovering sets over \(\mathbb{F}_{5^{6}}\), which are obtained from specific choices of subgroups in Theorem~\ref{lrc_two_max}.

\begin{example}\label{lrc_two_ex_2}
    Let \(q = 5^6=15625\). Then \(\sqrt{q} + 1 = 2 \cdot 3^2 \cdot 7\). By Theorem \ref{lrc_two_max}, we can construct locally repairable codes with two recovering sets in Table~\ref{LRC_two_5^6_1}. 
    \begin{table}[ht]
	\caption{Our locally repairable codes over $\mathbb{F}_{5^6}$} 
	\centering
	\begin{tabular}{|c|c|c|}
        \toprule[1.2pt]
		Locally repairable code \(\mathcal{C}\) & Parameters & Singleton-defect \( \Delta(\mathcal{C}) \) \\
        \midrule[1.2pt]
        
		\([15864,\,\ge1006,\,\ge 14004;(7,8)]_{5^6}\) & \makecell{\(h = 2\) and \(d_0 = 14001\) \\ \(|A_1| = 2\) and \(|A_2| = 3\)} & \(\le 0.044945\) \\
        \midrule[0.5pt]

        \([15864,\,\ge 1392,\,\ge 12528;(11,4)]_{5^6}\) & \makecell{\(h = 2\) and \(d_0 = 12525\) \\ \(|A_1| = 3\) and \(|A_2| = 2\)} & \(\le 0.112708\) \\

        \midrule[0.5pt]

        \([15582,\,\ge 3090,\,\ge 10878;(97,98)]_{5^6}\) & \makecell{\(h = 7\) and \(d_0 = 10789\) \\ \(|A_1| = 2\) and \(|A_2| = 3\)} & \(\le 0.101656\) \\

        \midrule[0.5pt]

        \([15582,\,\ge 777,\,\ge 13720;(146,49)]_{5^6}\) & \makecell{\(h = 7\) and \(d_0 = 13674\) \\ \(|A_1| = 3\) and \(|A_2| = 2\)} & \(\le 0.069247\) \\
        
        \bottomrule[1.2pt]
	\end{tabular}
	\label{LRC_two_5^6_1}
    \end{table}
    
    Note that the code lengths of the locally repairable codes in the last two examples of Table~\ref{LRC_two_5^6_1} are shorter than \( q \), but the locality property does not satisfy the conditions required by the result in \cite{8865660}. Therefore, the examples in the last two examples of Table~\ref{LRC_two_5^6_1} cannot be obtained from \cite{8865660} and are new results. For \(h = 3\), the condition \eqref{lrc_two_repair_condition} is not satisfied. However, by Remark \ref{lrc_two_re}(ii) and Sagemath, we can construct locally repairable codes with two recovering sets in Table \ref{LRC_two_5^6_2}. 
    \begin{table}[ht]
	\caption{Our locally repairable codes over $\mathbb{F}_{5^6}$} 
	\centering
	\begin{tabular}{|c|c|c|}
        \toprule[1.2pt]
		Locally repairable code \(\mathcal{C}\) & Parameters & Singleton-defect \( \Delta(\mathcal{C}) \) \\
        \midrule[1.2pt]
        
		\([15768,\,\ge4160,\,\ge 8964;(17,18)]_{5^6}\) & \makecell{\(h = 3\) and \(d_0 = 8947\) \\ \(|A_1| = 2\) and \(|A_2| = 3\)} & \(\le 0.152270\) \\
        \midrule[0.5pt]

        \([15768,\,\ge 1885,\,\ge 11691;(26,9)]_{5^6}\) & \makecell{\(h = 3\) and \(d_0 = 11683\) \\ \(|A_1| = 3\) and \(|A_2| = 2\)} & \(\le 0.134006\) \\
        
        \bottomrule[1.2pt]
	\end{tabular}
	\label{LRC_two_5^6_2}
    \end{table}

    With the help of Sagemath, we get the specific dimension of the first example in Table~\ref{LRC_two_5^6_1}. Specifically, we obtained a \([15864,\,1083,\,\ge 14004;(7,8)]_{5^6}\) locally repairable codes with two recovering sets. It means that the specific dimensions of the locally repairable codes we construct slightly exceed the lower bound established in Proposition~\ref{cons_lrc_two}. 
\end{example}

\subsection{A Construction over Finite Field of Characteristic 2}

Due to the inherent binary operations at the hardware level, codes over finite fields of characteristic \(2\) exhibit both extremely low implementation complexity and high error correction efficiency, forming the mathematical foundation for most practical coding schemes in modern storage and communication systems. Building on this advantage, we present a special construction of locally repairable codes with two recovering sets on a specific elliptic function field defined over a finite field of characteristic \(2\), which provides more locality that differs from that in Theorem~\ref{lrc_two_max}. 

Let \( q \) be an odd power of 4, i.e., \( q = 4^{2a+1} \) for an integer \( a \in \mathbb{N} \). Consider the elliptic function field \( E = \mathbb{F}_q(x, y) \) defined by the equation \( y^2 + y = x^3 \). By \cite[Lem. 15]{Optimal_Locally_Repairable_Codes_Via_Elliptic_Curves}, \( E/\mathbb{F}_q \) is a maximal elliptic curve. Let \( Q = (0, 1) \) and 
\[
    H = \langle Q\rangle = \{Q,[2]Q = (0,0),\mathcal{O}\}. 
\]
Let \( T_H = \{\tau_P:P\in H\} \) be a subgroup of \( T_E \). Let \(A_1\) be the cyclic subgroup of \(\operatorname{Aut}(E, \mathcal{O})\) generated by \(\sigma_1\) where \(\sigma_1\) is defined by \(\sigma_1(x) = x\) and \(\sigma_1(y) = y+1\). Let \(A_2\) be the cyclic subgroup of \(\operatorname{Aut}(E, \mathcal{O})\) generated by \(\sigma_2\) where \(\sigma_2\) is defined by \(\sigma_2(x) = u^2x\) and \(\sigma_2(y) = y\) with \(u\in\mathbb{F}_q\) satisfying that \(u^3 = 1\). Then we have \(|A_1| = 2\), \(|A_2|  = 3\) and \(|A_1\cap A_2| = 1\). It is easy to verify that \(\sigma(Q^\prime)\in H\) for any \(\sigma\in A_1A_2\) and \(Q^\prime\in H\). Therefore, by Lemma~\ref{sub}, we have the fact that \(T_HA_1A_2\) is a subgroup of \( \operatorname{Aut}(E/\mathbb{F}_q) \). Similar to the proof of Theorem~\ref{lrc_two_max_restate}, by
\[
    |T_HA_1A_2| = 18>|E[|H|]|=9,
\]
we can conclude that for any rational place \(Q\) of \( E^{T_HA_1A_2} \) that split completely in \( E \), the rational place $P$ of $E$ lying over $Q$ satisfies that \(P\notin E[|H|]\). Thus we can obtain \(q\)-ary \([n, k, d;(5,6)]_q\) and \([n, k, d;(8,3)]_q\) locally repairable codes by Proposition~\ref{cons_lrc_two}. 

\begin{theorem}\label{lrc_two_F2}
    Let \( q \) be an odd power of \( 4 \), i.e., \( q = 4^{2a+1} \) for an integer \( a \in \mathbb{N} \). For any integers \( d_0 \) with \( 1 \leq d_0 < n \) and \( m \) with \(1 \leq m \leq \left\lceil \frac{q+2\sqrt{q}-24}{18} \right\rceil\),  
    \begin{itemize}
        \item[(i)] there exists a \( q \)-ary \([n = 18m, k, d\ge n - L;(5,6)]_{q}\) locally repairable code with 
        \[
            k \geq 5\bigg\lfloor\frac{n-d_0}{6}\bigg\rfloor + 6\bigg\lfloor\frac{n-d_0}{9}\bigg\rfloor  + 2 - L, 
        \]
        \item[(ii)] there exists a \( q \)-ary \([n = 18m, k, d\ge n - L;(8,3)]_{q}\) locally repairable code with 
        \[
            k \geq 8\bigg\lfloor\frac{n-d_0}{9}\bigg\rfloor + 3\bigg\lfloor\frac{n-d_0}{6}\bigg\rfloor  + 2 - L,
        \]
    \end{itemize}
    where \(L = \max\left\{6\left\lfloor\frac{n-d_0}{6}\right\rfloor,9\left\lfloor\frac{n-d_0}{9}\right\rfloor\right\}\). 
\end{theorem}

\begin{remark}
    Let \( q \) be an odd power of \( 4 \). In the above construction, we can see that \(T_HA_1A_2\) is a subgroup of \( \operatorname{Aut}(E/\mathbb{F}_q) \) with order \(18\). Therefore, combining the construction framework in Section~\ref{lrc_framework}, we can provide an optimal \( q \)-ary \([18m, 17t+1, 18(m-t)]_{q}\) locally repairable code with locality \( 17 \) for any integers \( t \) and \( m \) satisfying \(1 \leq t < m \leq \left\lceil \frac{q+2\sqrt{q}-24}{18} \right\rceil \). Note that Theorem~1.1 of \cite{MA2023105686} also yields optimal locally repairable codes with locality \( r = 17 \) over \(\mathbb{F}_q\). However, the upper bound on \(m\) in our construction is slightly larger than that in the construction of Theorem~1.1 in \cite{MA2023105686}, which implies that our code length will be slightly longer.
\end{remark}

\section{Conclusion and Remarks}\label{conclusion}

The contributions of this paper are summarized as follows.
\begin{itemize}
    \item We present an alternative method for determining the functions \( e_i \) in the construction of locally repairable codes (see Proposition~\ref{find_e_i}). Our method does not rely on a set of completely splitting places. As a result, in the subsequent construction of locally repairable codes, we avoid using the modified algebraic geometry code technique compared to the previous constructions. 
    \item By employing the elliptic curves introduced in Proposition~\ref{new_ellp_cuv}, we construct several new families of optimal (classical) locally repairable codes (see Theorems~\ref{opt_LRC_one_hA} and \ref{opt_LRC_one_2h}). As shown in Example~\ref{opt_lrc_one_ex}, our results achieve two improvements: we produce optimal locally repairable codes over more finite fields (such as \(\mathbb{F}_{13^4}\) and \(\mathbb{F}_{37^2}\)), and offer more flexible locality. 
    \item We provide a general framework for constructing locally repairable codes with two recovering sets via elliptic function fields (see Proposition~\ref{cons_lrc_two}). Subsequently, using this framework and utilizing maximal elliptic function fields and Proposition \ref{new_ellp_cuv}, we construct several distinct families of locally repairable codes with two recovering sets (see Theorems~\ref{lrc_two_max}, \ref{lrc_two_nonmax} and \ref{lrc_two_F2}). Moreover, we can construct a family of \( q^2 \)-ary \([n, k, \ge n - 3\ell;(3\ell-1,\ell)]_{q^2}\) locally repairable code \(\mathcal{C}\) with length \(O(q^2 + 2q)\) and Singleton-defect \(\Delta(\mathcal{C})  = O\!\left(\frac{2\ell}{q^2 + 2q - 8\ell}\right) \) where \(\ell \mid\mid q + 2\) with \(4\ell<q\) (see Example~\ref{lrc_two_ex_1}). 
\end{itemize}

The lower bound on the dimension of the locally repairable codes with two recovering sets given in Proposition~\ref{cons_lrc_two} is actually relatively rough. As seen in Example~\ref{lrc_two_ex_2}, the specific dimensions of the locally repairable codes we construct slightly exceed the lower bound established in Proposition~\ref{cons_lrc_two}. Therefore, finding a sharper lower bound on the dimension presents an interesting problem. Besides, it would be also very interesting to consider how to construct locally repairable codes with more than two recovering sets via elliptic function fields. 


\section*{Acknowledgement}

This work is supported by Guangdong Basic and Applied Basic Research Foundation (No.~\seqsplit{2025A1515011764}), the National Natural Science Foundation of China (No. 12441107) and  the National Key Research and Development Program of
China (No.~\seqsplit{2025YFA1017100}).
\bibliographystyle{ieeetr}
\bibliography{reference_newmain}

\begin{IEEEbiographynophoto}{Junjie Huang}
received the BS degree from Sun Yat-sen University, Guanggdong, P.R. China. He is currently working toward the PhD degree in mathematics at Sun Yat-sen University. His research interests include function fields theory and algebraic coding theory.
\end{IEEEbiographynophoto}

\begin{IEEEbiographynophoto}{Chang-An Zhao}
received the bachelor’s degree in electronical engineering, the master’s degree in applied mathematics, and the PhD degree in information science and technology all from Sun Yat-sen University, Guangzhou, P.R.China, in 2001, 2005, and 2008, respectively. He works with the School of Mathematics, Sun Yat-sen University, Guangzhou, China. His research mainly focuses on elliptic curve cryptography, post-quantum cryptography and algebraic coding theory.
\end{IEEEbiographynophoto}

\end{document}